\documentclass[11pt]{article}
\usepackage[margin = 1.2 in]{geometry}
\usepackage[utf8]{inputenc}
\usepackage{authblk}
\usepackage{amsfonts}
\usepackage{amsmath}
\usepackage{amssymb}
\usepackage[english]{babel}
\usepackage{amsthm}
\usepackage{tikz}
\usepackage{float}
\usepackage[export]{adjustbox}
\usepackage{subcaption}
\usepackage{bbm}
\usepackage{comment}
\usepackage{natbib}
\usepackage{enumerate}
\usepackage{dcolumn}
\usepackage{lscape}
\usepackage{hyperref}
\usepackage{cleveref}
\setcitestyle{authoryear,open={(},close={)}}

\newcommand{\vst}{\vspace{3mm}}

\excludecomment{note}

\newtheorem{proposition}{Proposition}
\newtheorem{lemma}{Lemma}

\newtheorem{corollary}{Corollary}

\theoremstyle{remark}

\DeclareMathOperator*{\argmin}{argmin}

\newcommand{\und}{\underline}

\title{Redistribution Through Tax Relief\thanks{I am grateful to Marciano Siniscalchi, Nicola Persico, Joseph Stiglitz, Alessandro Pavan, Eddie Dekel, Bruno Strulovici, Francisco Poggi, Ludvig Sinander, and Melissa Figueira for illuminating conversations, and to seminar participants at Northwestern University for helpful comments.}}
\author{Quitz\'{e} Valenzuela-Stookey\thanks{Department of Economics, Northwestern University.}}
\date{November 3, 2020}

\begin{document}
\maketitle

\begin{note} 
    \Large \textcolor{blue}{WITH NOTES}
\end{note}

\vspace{-2em}
\begin{center}
    \Large \textcolor{blue}{\href{https://northwestern.box.com/s/lfbjcxtwyth5dcizqynqybyf31fjw9l2}{Click here for the latest version}}
\end{center}
\vst

\begin{abstract}
    This paper studies politically feasible policy solutions to inequities in local public goods provision. I focus in particular on the entwined issues of high property taxes, geographic income disparities, and inequalities in public education prevalent in the United States. It has long been recognized that with a mobile population, local administration and funding of schools leads to competition between districts. By accounting for heterogeneity in incomes and home qualities, I am able to shed new light on this phenomenon, and make novel policy recommendations. I characterize the equilibrium in a dynamic general equilibrium model of location choice and education investment with a competitive housing market, heterogeneous wealth levels and home qualities, and strategic district governments. When all homes are owner-occupied, I show that competition between strategic districts leads to over-taxation in an attempt to attract wealthier residents. A simple class of policies that cap and/or tax the expenditure of richer districts are Pareto improving, and thus politically feasible. These policies reduce inequality in access to education while increasing expenditure for under-funded schools. Gains are driven by mitigation of the negative externalities generated by excessive spending among wealthier districts. I also discuss the policy implications of the degree of homeownership. The model sheds new light on observed patterns of homeownership, location choice, and income. Finally, I test the assumptions and implications empirically using a regression discontinuity design and data on property tax referenda in Massachusetts.
\end{abstract}

Many schools in the United States are severely underfunded, and there is significant inequality in school quality across school districts.\footnote{See \cite{national1999equity}.} Academics and policy makers have long recognized that both of these issues are driven in large part by the existing system of school funding; in many states the majority of school funding comes from local sources, usually in the form of property tax revenue.\footnote{See for example \cite{kenyon2007property}.} Geographic disparities in income and home values are therefore reflected in school funding levels. These inequities are reinforced, moreover, but the sorting of wealthier residents into districts with better schools.\footnote{\cite{machin2011houses}.} 

Like school quality, the level of local property taxation is a salient political issue. Complaints of excessive property taxation are common, and efforts have been made by state governments across the country to provide property tax relief, with varying degrees of success.\footnote{\cite{mcguire2008local}.} The issues of school funding and property taxation are clearly linked, which suggests the need for a unified policy approach. The joint discussion of these issues is generally framed as a trade-off: without additional revenue sources, a district cannot have both well funded schools and low property taxes. Redistribution of resources from wealthy to poorer districts is an obvious solution to this dilemma, at least from the perspective of improving education in the neediest districts. However any such redistributive efforts must confront the tradition of local control in these areas. Existing geographic inequalities further complicate reform, as richer districts resist redistributive efforts from which they do not expect to benefit.\footnote{See for example \cite{klemens_2019}.} 

By accounting for the general equilibirum effects of tax and education policy, this paper offers a way around the school-quality/property-tax trade-off. The objective is a politically feasible policy to address the related problems of low school quality and high property taxes.\footnote{While I focus here on school funding in the U.S., the model is applicable to wide range of local public goods provision problems. For example, similar forces would arise in a model of developing countries offering incentives for foreign direct investment.} An ideal policy would \textit{i}) make use of existing policy tools, \textit{ii}) require no new revenue sources, \textit{iii}) provide property tax relief and improve school quality in the worst-funded districts, and \textit{iv}) make no districts worse off. Surprisingly, all of these criteria can be satisfied using a simple policy. The general solution identified in this paper is to cap expenditure on schools in the richest districts. Caps on taxation and/or expenditure, in various forms, are widespread. I show how, when correctly structured, caps on expenditure can provide welcome tax relief to the richer districts, while improving school quality in districts that are less well off. Over time, such a policy has the added benefit of reducing geographical income disparities. I explore variations on this policy. If monetary transfers can be made between districts it is possible to improve school quality at the bottom while providing property tax relief to \textit{all} districts. This can be done by charging a fee to districts that raise excessive revenue, rather than imposing a hard cap, and transferring the money raised in this way to poorer districts. Such a policy may be necessary to improve welfare when poor districts have a high proportion of renters. This type of ``tax-on-tax'' policy is in fact in place in Vermont.\footnote{\cite{saas2007school}.} This paper provides a novel justification for such policies, and suggests ways in which they can be refined. 

In contrast to the literature on redistribution \textit{within} districts with a mobile population (see for example \cite{epple1991mobility}), which in general finds that that population mobility limits the feasible degree of local redistribution, I show how population mobility can be used to facilitate redistribution \textit{in the population as a whole}. At a high level, the intuition for this result is simple. Consider first the case in which there are only two school districts, $A$ and $B$. I show, unsurprisingly, that in equilibrium the district with higher quality schools will attract a richer subset of the population, as the rich will be able to outbid the poorer for homes in these districts. I show that a richer population of new arrivals to the district will lead to higher home values, benefiting the current homeowners. Local governments are strategic. They recognize the link between school quality and the wealth of new arrivals. When most homes in the district are owner-occupied, this provides an incentive for them to increase school expenditure.\footnote{There are many additional reasons why local governments might prefer a richer population. Richer residents could contribute to higher sales tax revenue, or current residents may have an explicit preference for richer neighbors. These factors only reinforce the effect that I identify.} (On the other hand, when the ratio of homeowners to renters is sufficiently low the converse will hold; districts will spend inefficiently little on public goods. This has interesting implications, which I discuss further in Section \ref{sec:renters}.) However increases in a given district's school expenditure impose a corresponding negative externality on other districts, which now receive a poorer population. The result is inefficiency; in both districts taxes and expenditure are higher in equilibrium than would be optimal given the equilibrium income distribution within each district.\footnote{I will sometimes refer to this phenomenon as ``excessive expenditure'', but it should be emphasised that ``excessive'' is meant in a very limited sense; given the equilibrium income distribution in a district, the district would be better off with a lower expenditure level. This is not to say that the level of education spending is excessive in a normative sense. Indeed, the policy objective is to \textit{increase} education spending in the poorest districts.} There exist binding expenditure caps such that the set of residents in each district remains the same when the caps are imposed, and which induce a better outcome from the perspective of both districts. This is possible because the caps reduce the negative externalities that districts impose on each other. 

Suppose now that there is a third district, $C$, with lower school quality than $A$ and $B$. It is still the case that there are caps on expenditure in $A$ and $B$ that reduce negative externalities and makes both $A$ and $B$ better off, despite the fact that some rich residents will choose to move to $C$ given the lower school quality in $A$ and $B$. But $C$ benefits from the richer new residents and correspondingly higher property values, which allows it to spend more on schools. 

Throughout, I refer to the public good being provided by school districts as education. In reality education may be just one of many goods, such as the number and maintenance of  parks and cleanliness of downtown areas, which together determine the perceived quality of a district. Such an interpretation is perfectly consistent with the model. The excessive expenditure caused by inter-distinct competition may take the form of superfluous decoration of public buildings or landscaping of business districts, intended to attract wealthier residents and support local home values. 

While I focus here on residential property taxation and location choice, similar forces are at play in other settings, such as the market for higher education in the United States. Increasing competition between colleges in the U.S. has been associated with higher tuition levels \citep{hoxby2000effects}. Popular commentary has attributed this increase, at least in part, to expenditures on extravagant amenities such as lazy rivers and opulent dinning halls \citep{horn_2019}. While it has been widely recognized that this phenomenon is driven by competition for students, the role of income heterogeneity has received less attention. This paper suggests that such spending may not be the result of efficient Tiebout competition. Rather, it is induced by the income sorting effect described above. 

This paper falls within the large literature on mobility and local public goods stemming from the seminal paper of \cite{tiebout1956pure}, which argues that by ``voting with their feet'', migrants induce efficient provision of local public goods. In contrast, my model predicts inefficiently high expenditure. This finding relates to the strand of literature focusing on the fiscal externalities of migration. Such externalities, as studied by \cite{buchanan1972efficiency}, \cite{flatters1974public} and \cite{stiglitz1982theory}, generally lead to inefficient equilibria. \cite{stiglitz1982theory} discusses the case in which a majority of residents rent, and also finds that this leads to under-provision of the local public good. \cite{starrett1980method}, and the subsequent comment by \cite{boadway1982method} highlight the importance of understanding migration responses to taxation and public goods quality when evaluating efficiency of public goods provision. To my knowledge, the literature has not explored the use of tax caps as a way to \textit{increase} public goods provision in a subset of districts. 

The key features of the model presented here are heterogeneous wealth across individuals, heterogeneous home qualities, and strategic choice of local public goods provision by districts. Models with one or another subset of these features have been studied in the literature, but no paper has, to my knowledge, incorporated them all. Each plays an important role, as we will see, in the substantive implications of the model. 

\vst
\noindent\textit{Heterogeneous wealth}

The strategic choice by districts of the level of expenditure on the public good is an important component of my analysis. \cite{fernandez1996income} also evaluates policies designed to improve school quality in a model with location choice, but with expenditure chosen after location choices have been made. In their model an expenditure cap on the wealthy districts induces a Pareto \textit{inferior} outcome. \cite{brueckner2001local} find empirical evidence of strategic behavior in the choice of property tax rates by local governments, suggesting that such behavior should not be ignored. 

\vst
\noindent\textit{Heterogeneous wealth, strategic districts}

\cite{epple1991mobility} study a model model of redistribution by local jurisdictions. Their model shares some features with the current paper, most importantly heterogeneous wealth levels and strategic districts that recognize the migration effects of their tax policies. The key conceptual difference however is that tax revenue is used to make cash transfers, as opposed to providing a local public good. This has important implications; in their model individuals also sort based on income, but high income individuals go to the districts with the lowest levels of taxation (and redistribution). In contrast, since in my model all individuals derive the same marginal utility from public goods consumption, high income individuals sort into the districts with the highest taxes (and levels of public goods). This is the determinant of the strategic behavior of districts. 

\vst
\noindent\textit{Heterogeneity of wealth and home quality}

\cite{de2003equilibria} study heterogeneous wealth in a spacial model in which location quality is determined by the distance to the city center. As in the current paper, they find that complete sorting of households according to wealth does not occur, in contrast to models with homogeneous home quality such as \cite{epple1991mobility}. Unlike the current paper however, taxation levels are chosen by myopic voters, thus eliminating they key strategic considerations that lead to over-taxation.

\vst
There is a large empirical literature studying the role of school quality in location choice. \cite{black1999better} identifies a significant causal effect of school quality on home values. \cite{kane2006school} and \cite{bayer2004tiebout} obtain similar results, while also emphasising the importance of individuals preferences with regards to their neighbors in reinforcing stratification and contributing to higher home prices in districts with good schools. Such preferences only reinforce the mechanism which leads to inefficiently high expenditure in my model, and increases the scope for policy intervention. \cite{bayer2019efficient} directly study the effect of school expenditure and property taxes on home prices, finding a significant positive effect. 

My analysis highlights the importance of two phenomena which have received limited attention in the empirical literature. First, the key determinant of both housing prices and location choice is not the level of school quality, but the differences in school quality across districts. The effect of school quality and tax levies on home values in neighboring districts has received surprisingly little attention in the empirical literature. Empirical strategies aimed at identifying causal effects of school quality on home prices and migration should focus on differentials in school spending, rather than levels. This suggests that empirical strategies using school assignment zone boundary discontinuities, such as \cite{black1999better} and \cite{kane2006school}, which exploit differences in school quality across districts, are more likely to give reliable estimates than those, such as \cite{bayer2019efficient}, that exploit exogenous variation in expenditures over time. If neighboring districts experience expenditure shocks simultaneously, such a strategy is likely to underestimate the true effect of school quality on home prices. 

Second, changes in local income distributions induced by changes in school quality differentials are the key driver of local property values and of the choice of expenditure levels by local governments. The simultaneous causality between local income levels and school quality has long been recognized. However empirical studies of the causal effect of schools on the income distribution are lacking. Empirical work should also seek to understand the strategic consideration of this effect by local governments.

I attempt to partly fill the gap in the empirical literature, while also testing the main assumptions and predictions of the model. To identify the causal impact of property taxes, I exploit municipal property tax caps in Massachusetts. State law caps the total property tax levy for each municipality, but allows municipalities to override this cap through a majority vote in a referendum. I use a regression discontinuity design, exploiting the discontinuity in outcomes at the 50\% vote share, to identify the effect of increases in the property tax levy on property values and incomes, both within a given municipality and in neighboring municipalities. I find strong evidence that an increase in the property tax levy raises home values within a given municipality. This is the key driver of the over-taxation result. Surprisingly, an increase in the tax levy in one municipality also raises property values in neighboring municipalities, even when these municipalities do not share a school district. Thus suggests that there are spillovers across districts that are not accounted for in the data. Additionally, I find some direct evidence of income sorting; an increase in the tax levy raises the average income in a municipality. 

The remainder of the paper is organized as follows. Section \ref{sec:model} presents the model, Section \ref{sec:housing_mkt} characterizes the equilibrium in the housing market, Section \ref{sec:game} formalizes the strategic interaction between districts, presents a key comparative statics result, characterizes equilibrium in the game between districts, and discusses the policy recommendations. Section \ref{section:empirical} contains the empirical analysis. Omitted proofs, alternative approaches, and supplementary tables are contained in the appendix.  

\section{The model}\label{sec:model}
For ease of exposition, I will first present the model and results with only two districts. The generalization to more than two districts, necessary for the central policy recommendation of the paper, is straightforward and will be discussed in Section \ref{sec:multiple}. Label the two districts $A$ and $B$. Each district has a school and a unit mass of houses.\footnote{The assumption of perfectly inelastic housing supply is made for tractability. The important implication of this is that taxes are fully capitalized into home prices. However the central intuition of the paper, that income sorting incentivizes over-taxation, continues to hold when supply is elastic. Very low elasticities of housing supply are observed in land-constrained metropolitan areas (see for example \cite{saiz2010geographic}). This assumption also highlights the distortions arising from population mobility: with perfectly inelastic housing supply, taxes are non-distortionary in the traditional sense. Nevertheless, tax caps can be Pareto improving, highlighting the migration externality that is the focus of this paper.} House quality is denoted by a number in $[0,1]$. The distribution of house quality in district $j$ is described by the CDF $Q^j$ with support $[\und{q}^j, \Bar{q}^j]$. Assume that $Q^j$ is differentiable with strictly positive derivative on its support (including 1-sided derivatives at boundaries) for $j \in  \{A,B \}$. Schools are financed by property tax revenue. School quality in district $j$ is given by $s(e^j)$, where $s$ is an increasing concave function and $e^j$ is expenditure in district $j$.\footnote{The technology converting expenditure to school quality need not be the same in all districts; all results continue to hold if school quality is given by $s^j(e^j)$ for $j \in \{A,B \}$. Such differences may reflect different needs of the student population, and should be taken into account when designing tax cap policies.} Districts choose tax schedules $\tau^A, \tau^B: \mathbb{R} \mapsto \mathbb{R}$, which give tax liability as a function of home value. School expenditure must be fully funded by local tax revenue.\footnote{Allowing for fixed subsidies from other sources does not change the characterization of equilibrium, but the dependence of subsidies on expenditure levels may matter for policy recommendations, as I will discuss below.}

Individuals are modeled with overlapping generations. Individuals live for two periods.\footnote{All results extend immediately to individuals living for any number of periods.} Each individual is characterized by their per-period income $w \in [\und{w},\bar{w}]$, which I refer to as their type. In the first period of their life, individuals purchase a house, enjoy the schools in their district of residence, and engage in consumption. In the second period they sell the house and engage in consumption. Individuals can borrow and save at rate $r$. Suppose that in their first period of life a type $w$ individual purchases a house of quality $q$ in district $j$ where education expenditure is $e^j$, paying a price $p_1$ and tax bill $\tau$, and saving an amount $b$. In the second period the individual sells the house at price $p_2$. Their total payoff is given by 

\begin{equation}
    q + s(e^j) + u(w - p_1 - \tau - b) + u(w + p_2 + (1+r)b). 
\end{equation}
where $u$ is continuously differentiable and strictly concave. Here $w - p_1 - \tau - b$ is first period consumption and $w + p_2 + (1+r)b$ is second period consumption. 

At the start of each period a new population of individuals with mass 1 arrives to replace the outgoing generation. Let $F(w)$ be the CDF of wealth levels in the incoming population. Assume that $F$ is differentiable with strictly positive derivative on its support $[\und{w},\bar{w}]$ (including 1-sided derivatives at boundaries). I refer to a \textit{location} as a pair $(q,j)$ denoting a home quality and district. Individuals are free to choose any location in the two districts. 

\section{Housing market equilibrium}\label{sec:housing_mkt}

The housing market is competitive. I will first consider an intermediate notion of equilibrium in the housing market, in which I ignore the constraint that tax revenue must be sufficient to fund the districts' chosen expenditure levels, as well as the strategic choice of these tax schedules and expenditure levels by the district governments in current and future periods. That is, I take the tax schedules and expenditure levels in each district as given, and study the equilibrium in the housing market. 

Fix the tax schedules and expenditure levels in each district. An equilibrium consists of price functions $p_t^j:[\und{q}^j, \bar{q}^j] \mapsto \mathbb{R}$ for $j \in \{A,B\}$ and $t \in\{1,2 \}$ such that markets clear given that all individuals choose their location optimally, taking prices and induced tax bills as given. Market clearing in this setting means that the mass of individuals moving to homes of quality less than or equal to $q$ in district $j$ is exactly $Q^j(q)$.\footnote{Throughout, statements about equilibrium allocations are subject to the caveat ``except for zero measure sets'', which I will not state.} Given price functions, I will sometimes abuse notation and write the tax bill as a function of the location, rather than the home price. 

Given  price functions, a location $(q,j)$ is characterized by home quality $q$, school quality $s(e^j)$, buying price $p_1^j(q)$, tax bill $\tau^j(q)$, and selling price $p_2^j(q)$. For a given \textit{payment-price vector} $(\tau^j(q), p_1^j(q), p_2^j(q))$, define the present discounted value (PDV) as
\begin{equation*}
    m^j = -p^j_1(q) - \tau^j_1(q) + \dfrac{1 }{1+r}p_2^j(q).
\end{equation*}
Define the \textit{location quality} of $(q,j)$ as $s(e^j) + q$. Given a payment-price vector $z = (\tau, p_1, p_2)$ define the money value function for type $w$ as 
\begin{equation}\label{eq0.1}
    V(w,z) = \max_{b \in \mathbb{R}} \ u\left(w - p_1 - \tau - b \right) +  u\left(w + p_2 + (1+r)b \right)
\end{equation}

Let $b(w,z)$ be the solution to this problem. The ranking on payment-price vectors induced by $V$, i.e. by optimal borrowing/saving, is the same for all types. In particular, payment-price vectors are ranked according to their PDVs.  When convenient, I will abuse notation and write $V(w,m)$ rather than $V(w,z)$, where $m$ is the PDV of $z$. In equilibrium it obviously must be the case that higher quality locations come with lower PDVs, otherwise no types would want to move to the low quality location, violating market clearing. 

Say that an equilibrium is \textit{monotone} if higher types always move to higher quality locations than lower types. The following monotonicity property greatly simplifies equilibrium characterization. It follows directly from concavity of $u$, which implies that higher income individuals have a higher marginal willingness to pay for location quality. The following lemma simplifies the statement of the market clearing condition.

\begin{lemma}\label{lem0.1}
    Any equilibrium is monotone. 
\end{lemma}

The power of Lemma \ref{lem0.1} is that there is in fact a unique allocation of individuals that satisfies monotonicity and fills all homes. This is best visualized by allocating homes from the bottom up, starting with the lowest-type individuals. Let $\Delta = s(e^A) - s(e^B)$ be the school quality gap between $A$ and $B$, and suppose $s(e^B) + \und{q}^B < s(e^A) + \und{q}^A$, so the worst quality locations are all in district $B$. Then monotonicity requires that the lowest types move to district $B$. This will continue until homes in district $B$ have been filled up to those with quality $q_* := \und{q}^A + \Delta $. Home $q^*$ in district $B$ has the same location quality as home $\und{q}^A$ in district $A$, the higher home quality exactly compensating for the lower school quality. At this point individuals of a given type will be split between districts, some individuals of a given type will more to district $A$ and some to district $B$, according to the relative availability houses in the two districts.\footnote{This description is of course informal given that there are a continuum of individuals.} Regardless of which district they move to however, they will have the same location quality. Eventually one of the districts will be full, and the remaining individuals, those with the highest types, will go to the other district. Formally, the equilibrium location allocation can be described as follows.

\begin{proposition}\label{prop1}
  Fix $e^A, e^B$, and without loss of generality let $B = \argmin_{j} \und{q}^j + s(e^j)$. The unique equilibrium location allocation takes the following form:
  \renewcommand{\labelenumi}{\roman{enumi}.}
  \begin{enumerate}
    \item Let $q_* = s(e^A) - s(e^B) + \und{q}^A$. Let $w_*$ satisfy $F(w_*) = Q^B(q_*)$. All individuals with $w < w_*$ locate in district $B$, with housing quality increasing in $w$. 
      
    \item If $\bar{q}^A + s(e^A) \geq \bar{q}^B + s(e^B)$, let $q^*$ satisfy $q^* + s(e^A) = \bar{q}^B + s(e^B)$ and let $w^*$ satisfy $F(w^*) = Q^A(q^*) + 1/2$. All individuals with $w > w^*$ locate in district A. If $\bar{q}^A + s(e^A) < \bar{q}^B + s(e^B)$, let $q^*$ satisfy $\bar{q}^A + s(e^A) = q^* + s(e^B)$ and let $w^*$ satisfy $F(w^*) = Q^B(q^*) + 1/2$. All individuals with $w > w^*$ locate in district B. Housing quality is increasing in $w$.
\end{enumerate}
\end{proposition}

It will convenient to work with functions $q^j: [\und{w},\bar{w}]: \mapsto [\und{q}^j,\bar{q}^A]$ which describe the home that a type-$j$ individual would move to in district $j$. Formally, let $\Gamma^j: [\und{w},\bar{w}] \mapsto [0,1]$ be the mass of new individuals with type less than or equal to $w$ moving to district $j$, which fully describe the location allocation given Lemma \ref{lem0.1}. The market clearing requirement that all individuals find a house can be written as
\begin{align}
    \Gamma^A(w) + \Gamma^B(w) &=  F(w) \ \ \forall \ w \in W \label{eq:marketclearing1}\\
    \Gamma^A(\Bar{w}) &= \Gamma^B(\Bar{w}) = \frac{1}{2}.\label{eq:marketclearing2}
\end{align}
Define the function $q^j(w)$ implicitly by
\begin{equation}\label{eq:location}
    Q^j(q^j(w)) = \Gamma^j(w).
\end{equation}
So $q^j(w)$ is the $\Gamma^j(w)$ quantile of $Q^j$. By Lemma \ref{lem0.1} home quality choices within each district will be monotone in type, so we can also refer to $q^j(w)$ as the district $j$ \textit{house choice} for type $w$. These are determined by (\ref{eq:location}) for all $w$ that move to district $j$ in equilibrium. Conversely, given $q^A,q^B$ we can back out $\Gamma^A$ and $\Gamma^B$ via (\ref{eq:marketclearing1}),(\ref{eq:marketclearing2}), and (\ref{eq:location}).

Proposition \ref{prop1} used the fact that optimal location choice implies monotonicity to pin down the equilibrium location allocation. In order pin down equilibrium home prices, we need to consider more carefully the incentives of the new arrivals. From the perspective of an individual, each location is characterized by a location quality and payment-price vector (or, equivalently, a PDV). For any given prices and tax schedule, the location allocation therefore specifies a payment-price allocation $z^j(w)$ with PDV $m^j(w)$. Under market clearing, an allocation of individuals satisfies individual optimality if and only if it satisfies an incentive compatibility constraint; no individual strictly prefers the location choice of another individual.\footnote{This equivalence depends crucially on the fact that taxes are based on location, wealth.} 

Given a location allocation, expenditure levels, tax schedules, and prices, let $U(w)$ be the value function for type $w$:
\begin{equation}
    U(w) = \max_{w'\in [\und{w},\bar{w}], \ j \in\{A,B\}} s(e^j) + q^j(w') + V(w,m^j(w')).
\end{equation}
Equivalently, $U(w) = \max\{U_A(w), U_B(w)\}$, where $U_j(w)$ is the value when $w$ is restricted to district $j$. Our goal is to characterize the equilibrium prices and location allocations. To do so, it is without loss to restrict attention to pricees that satisfy within-district IC: any type (weakly) prefers their own payment-price vector and house in district $j$ to that of any other type.\footnote{Within-side IC is not implied by IC. A type $w$ could prefer the district $A$ allocation of $w'$ to its own, as long as $w$ prefers its district $B$ allocation to both. But in such cases we can always modify $q^A, q^B$ to satisfy within-side IC without changing the actual location allocation. } By \cite{milgrom2002envelope}, if the mechanism is within-district IC then $U_j$ is absolutely continuous and satisfies the envelope formula.
\begin{equation}\label{eq0.8}
    U_j(w) = U_j(\und{w}) + \int_{\und{w}}^w V_1(x,z(x))dx \ \ \ \ \forall \ w.
\end{equation}
It will sometimes be helpful to separate the value function $U$ into its the monetary and non-monetary components. Let $M_j(w) = V(w,z^j(w))$, so that $U_j(w) = s(e^j) + q^j(w) + M_j(w)$.

The following lemma is an immediate consequence of continuous value functions is. It says that if a type strictly prefers a given district then there is a neighborhood of nearby types that also strictly prefers that district. 

\begin{lemma}\label{lem0.2}
If $U_A(w) > U_B(w)$ then there exists $\varepsilon$ such that for all $ w' \in (w-\varepsilon, w + \varepsilon)$, $U_A(w') > U_B(w')$. 
\end{lemma}
 
Lemma \ref{lem0.2} has the following immediate implication.
 
\begin{corollary}\label{cor:indifference}
Let $w^*,w_*$ be as in Proposition \ref{prop1}. All types $w \in [w_*, w^*]$ satisfy $q^A(w) + s(e^A) = q^B(w) + s(e^B)$ and $m^A(w) = m^B(w)$. Housing quality is increasing in $w$. 
\end{corollary}

Given the location allocation, the PVDs are determined via the envelope condition (\ref{eq0.8}). Proposition \ref{prop1}, and in particular point $iii.$, implies that we can decompose the individual's value function into the monetary and non-monetary components, without specifying their district choice. Thus write $U(w) = \ell(w) + M(w)$, where $\ell(w) = q^j(w) + s(e^j)$ when type $w$ chooses district $j$. I refer to $M$ as the money value.

Proposition \ref{prop1} is helpful for understanding migration patterns. It also greatly simplifies the game played by districts when choosing their expenditure levels and tax schedules. The location allocation is uniquely determined by the school quality gap between the two districts. This means that changes in the structure of the tax schedule, for example making property taxes more or less progressive, will not impact migration patterns (although they will affect welfare). This is due not to capitalization of taxes into the home price, but to the single-crossing property of preferences; richer individuals will always outbid others for the best locations. Since the location allocation is entirely pinned down by the school quality gap, the only strategic choice that the district must make is the level of expenditure. Given the expenditure levels, the district knows who will move where, and can structure the tax schedule to raise the desired revenue, without worrying about the impact this may have on location decisions. This separation will be used in Section \ref{sec:game} to analyse the equilibrium in the game played between districts, and derive the policy recommendations. The only missing piece is to understand equilibrium payment-price vectors. 

To pin down the equilibrium PDVs it only remains to determine the value of the lowest types. The envelope condition then fully determines all the PDVs. I assume that the money value of the lowest type, $M(\und{w})$, is fixed at some value $\und{M}$. This is due to the existence of subsidized housing.\footnote{Subsidized housing is just one explanation for the existence of a fixed outside option. It could also be the result of other forms of government intervention.} Assume that in each district $j$ there is a mass $\sigma> 0$ of tax exempt ``subsidized housing'' with quality $\und{q}^j$. Subsidized housing is available at a the same fixed, not market, price in both districts.\footnote{Really all that is needed is that there is subsidized housing available in whichever district the poorest new arrivals move to, as we will see.} Moreover, suppose that there is an additional mass $\rho$ of individuals with types continuously distributed on $[0,\und{w}]$, where $\sigma < \rho < 2\sigma$. Subsidized housing is allocated on the basis of income, with priority given to low income individuals. If type $\und{w}$ strictly prefers market housing in district $j$ to subsidized housing in district $j$ then so will some type below $\und{w}$, and thus, since $\rho > \sigma$, there will be excess demand for market housing. On the other hand, since $\rho < 2\sigma$, there will be vacant subsidized housing in district $j$ if $s(e^j) + \und{q}^j < s(e^{-j}) + \und{q}^{-j}$. Since these subsidized units are available, type $\und{w}$ must be indifferent between choosing the low quality market housing and the subsidized housing in whichever district $j$ satisfies $s(e^j) + \und{q}^j < s(e^{-j}) + \und{q}^{-j}$. Since subsidized housing is tax exempt and has a fixed price, this implies a fixed value for $M(\und{w})$.

Proposition \ref{prop1} implies that the location allocation is uniquely pinned down by the expenditure gap $s(e^A) - s(e^B)$. This is also enough to pin down the money value $w \mapsto M(w)$, given the fixed money value $\und{M}$ for type $\und{w}$. 

\begin{lemma}\label{lem3.5}
$M$ is uniquely determined by $s(e^A) - s(e^B)$
\end{lemma}

As long as there is a type that is indifferent between the two districts, there is a one-to-one mapping between the location allocation and the expenditure gap. This breaks down when one district's location quality dominates the other's; for example if $\und{q}^A + s(e^A) > \bar{q}^B + s(e^B)$. When this holds, small changes to either district's expenditure levels will have no impact on the location allocation. A corollary of Proposition \ref{prop1} and Lemma \ref{lem3.5} is that $M$ is uniquely determined by the expenditure gaps between the two districts, as long as there is no such dominance. 

\begin{corollary}\label{cor1}
Assume there exists a type that is indifferent between the two districts (i.e. there is no quality dominance as described above). Then $M$ is uniquely determined by the location allocations.
\end{corollary}

The following lemma, which will be useful later on, follows from a similar argument. It is not necessary to rule out quality dominance here, since the subsidized housing outside option pins down $M(\und{w})$.

\begin{lemma}\label{lem3.6}
Consider two equilibria in which all types in $[\und{w}, w']$ move to the same location. Then all types in this interval receive the same money values in the two equilibria.
\end{lemma}

Lemma \ref{lem3.5} has what at first glance appears to be a startling implication. Since the money values depend only on the expenditure gap, not on levels, starting from the current generation of incoming residents all individuals payoffs can be arbitrarily improved from a given equilibrium (assuming $s$ is strictly increasing) by fixing the expenditure gap in every period and taking both districts' spending levels to infinity (this requires negative home prices). There is no free lunch however. The increase in current and future school quality is paid for by current residents who are selling their homes. 

For example, consider two period model and a steady state equilibrium such that home prices are constant over time. Then the home price in a given location satisfies
\begin{equation*}
    p = -\tau + \dfrac{1}{1+r}p - m
\end{equation*}
so the price is given by 
\begin{equation*}
    p = -\dfrac{1}{1+r}\left( \tau + m  \right)
\end{equation*}
Suppose that in both districts expenditure in all periods increases by $\varepsilon$, and that this is funded by an increase in the tax bill of $\varepsilon$ for each home. Then the location allocation and lump-sum equivalents remain unchanged. If the economy remains in steady state then the new price is 
\begin{equation*}
    \hat{p} = -\dfrac{1}{1+r}\left( \tau + m + \varepsilon \right).
\end{equation*}
Thus the decrease in the home price following the increase in expenditure is $\frac{1+r}{r} \varepsilon$. This means that the current generation of home sellers creates an endowment which pays a dividend in each period sufficient to fund the increased school expenditure.

This observation implies that the difficult task, from a social welfare perspective, is to increase school quality without disproportionately harming the current generation of residents. I now address this objective. The first step is to understand the game played between the districts. Before formalizing the interaction between the two districts, I will make a few more observations regarding the competitive equilibrium outcome.

\subsection{Comparative statics}
I now consider comparative statics on the distribution of residents. This result is the driving force behind the welfare conclusions and policy recommendations of the remainder of the paper. It is relevant beyond the specific district objectives discussed in the next section.

Suppose that school quality in one district increases, holding that in the other fixed. I want to know what happens to payment-price vectors and the welfare of both new arrivals and existing residents. In any equilibrium with payment price vector $z^j$ in district $j$, let $z_q^j$ be the payment-price vector associated with home $q$ in district $j$, and let $w^j_q$ be the type moving home $q$ in district $j$. Let $m^j_q$ be the PDV function.

\begin{proposition}\label{prop2}
Consider expenditure profiles $\{e^A, e^B \}$ and $\{\hat{e}^A, \hat{e}^B \}$ with $\hat{e}^A - \hat{e}^B > e^A - e^B$ (without loss of generality let $\hat{e}^B = e^B$). Let $\hat{m},m$ be the associated PDV functions.  Then
\begin{itemize}
    \item If $\und{q}^A + s(e^A) \geq \bar{q}^B +s(e^B)$ then $\hat{m}^B_q = m^B_q$ and $m^A_q > \hat{m}^A_q$ for all $q$.
    \item If $\bar{q}^A + s(\hat{e}^A) < \und{q}^B + s(\hat{e}^B)$ then $\hat{m}^A_q = m^A_q$ and $\hat{m}^B_q > m^B_q$ for all $q$.
\end{itemize}
If neither of these hold then
\begin{itemize}
    \item If $\und{q}^A + s(\hat{e}^A) \geq \und{q}^B + s(\hat{e}^B)$ let $q^*$ satisfy $\und{q}^A + s(\hat{e}^A) = q^* + s(\hat{e}^B)$. Then $m_q^A > \hat{m}_q^A$ for all homes in district $A$. In district $B$ $\hat{m}_q^B > m_q^B$ for $q > q^*$ and $\hat{m}_q^B = m_q^B$ for $q \leq q^*$.
    \item If $\und{q}^A + s(\hat{e}^A) < \und{q}^B + s(\hat{e}^B)$ let $\hat{q}^*$ satisfy $\hat{q}^* + s(\hat{e}^A) = \und{q}^B + s(\hat{e}^B)$. Then $m_q^A > \hat{m}_q^A$ for all $q > q^*$ and $m_q^A = \hat{m}_q^A$ for all $q \leq q^*$. In district $B$ $\hat{m}_q^B > m_q^B$ for all $q$.
\end{itemize}
Moreover, whenever $m^A_q > \hat{m}^A_q$ holds, 
\begin{equation*}
 \limsup_{\hat{e}^A \searrow e^A} \dfrac{\hat{m}^A_q - m^A_q}{\hat{e}^A - e^A } < 0
\end{equation*}
and whenever $\hat{m}^B_q > m^B_q$, 
\begin{equation*}
 \limsup_{\hat{e}^A \searrow e^A} \dfrac{\hat{m}^B_q - m^B_q}{\hat{e}^A - e^A } > 0.
\end{equation*}

\end{proposition}

\section{District actions}\label{sec:game}

The equilibrium characterization says that if a district raises expenditure on schools, holding fixed that of the other district, it will attract a richer subset of the incoming individuals. This conclusion is independent of the tax schedules used to collect the additional revenue. There are many reasons to suspect that districts may prefer richer residents. Richer individuals may make private donations to help fund schools, or spend more at local businesses. I will focus on another consequence of having a richer population of newcomers, which is that it will generally lead to higher home prices, benefiting current residents.

Competition for richer residents will push a district to gather more tax revenue and invest in higher quality schools. Taxes, however, will be inefficiently high. To be precise, assume that the district prefers to attract a richer population of newcomers. Suppose a district's optimal level of revenue is $e$, which results in some population of both old and new residents. Then fixing the population, the district would prefer to lower $e$, even though this comes at the cost of lower quality schools. This inefficiency means that both districts may benefit from a cap on tax revenue.

\subsection{Two period government}\label{sec:district_game}
Let me now formalize the game played between the two districts. The main intuition, that the desire to attract richer residents leads to over-taxation, is easily seen in a simple model in which districts commit to the tax bills associated with each home. 

District governments are in office for two periods, and all districts are on the same election cycle. In what follows I consider the situation both district governments have just entered office, and refer to this period as period 1. At the start of period 1 the government of district $j$ commits to expenditure levels $e^j_1$, $e^j_2$ for each of the next two periods. The government also commits to the tax bills associated with each home. That is, it chooses $\tau_1^j, \tau_2^j$ where $ \tau_1^j: [\und{q}^j, \bar{q}^j] \mapsto \mathbb{R}$ specifies the tax bill for each home in period $i$.\footnote{Expenditure levels for period 2 may need to be determined in advance for planning purposes, or may reflect campaign promises. The fact that the government commits to tax bills for each home, rather than a more general function of home prices for example, is motivated by the fact that home value assessments occur infrequently.} Given the announced expenditure levels and tax schedules, competitive equilibrium payment-price vectors are determined for every house and individuals choose their location. 

I assume that the district government is an agent maximizing a welfare objective, which will be discussed below. Many papers in the literature assume that district decisions are made by majority rule. Reality is often somewhere in between; decisions are made by officials who face periodic elections. Modeling the district's decision through voting would not change the qualitative results, as the incentive to over-tax would still exist for the median voter. I do not model voting for two reasons. First, I want to emphasize that Parteo improvements are available even when the district government is maximizing welfare. Second, majority rule does not fit well with other applications, such as competition between universities. 

District governments and individuals must anticipate the actions of subsequent regimes, as these will impact current home prices. I will assume that the strategies of all regimes satisfy a natural \textit{Markov property}; they depend only on the wealth (income net borrowing and saving) of old residents in the first period. This is the only payoff relevant state variable. Under this assumption, it is natural to assume that individuals' beliefs about the savings strategies of other individuals are also Markov: they depend only on the current period wealth distributions and tax schedules in the two districts. These two assumptions are related. Individuals faced with the same payoff-relevant state expect others to behave in the same way, independent of history. Thus it is rational for each individual to follow a Markov strategy, since the wealth of all individuals in the subsequent period will be the same, again regardless of history, and thus the strategies of the subsequent regimes will be the same given the Markov assumption on governments.  

The objective of the district government can depend on expenditures, tax schedules, prices, and the wealth distributions of old and new residents, provided it is increasing in the welfare of old and new residents. For clarity, I focus on districts that care about school quality and the welfare of old residents, so their objective is given by 
\begin{equation*}
    \Pi^j(e^j_1, e^j_2, \tau^j_1, \tau^j_2) = \theta (s(e^j_1) + s(e^j_2)) + \int u(\tilde{w}^j_q + p^j(q)) dQ^j,
\end{equation*}
where $\tilde{w}^j_q$ is the wealth of old residents (their income net of borrowing and saving that they did in the previous period) and $\theta$ is the weight placed on school quality.\footnote{In the appendix I consider an alternative formulation of the district government objective, for which the two-period commitment assumption is not needed. The strength of the version presented here is that the utilitarian objective is intuitive and familiar.}   It is easy to see that if the district directly cares about the welfare of new residents it will want to attract a richer set. I want to show that the district has an incentive to over-tax even without this benefit. Furthermore, it is natural to think that the decisions of the government are shaped by old residents, if they are made between the departure of the previous generation and the arrival of new residents. 

Consider the savings problem faced by the new entrants in period 2. The income distributions of the new arrivals in the two districts is determined uniquely by $e^A_2$ and $e^B_2$. Furthermore, given the Markov assumptions on individual strategies, their savings behavior will be independent of the actions taken in period 1. Thus the wealth distribution of these residents in the subsequent period, the first of the new regime, will also be independent of period 1 actions. By the Markov assumption on government strategies, the continuation game will therefore be independent of period 1 actions. This observation is summarized in the following lemma.

\begin{lemma}\label{lem5.7}
Fix the period 2 tax schedules (and thus the expenditure levels) of both districts. The actions of all individuals (location and savings choices) and regimes (taxation and expenditure choices) from period 2 on are independent of the actions taken in period 1.
\end{lemma}

Lemma \ref{lem5.7} makes it easy to study the problem of the government choosing it's period 1 tax regime. It implies that home prices in period 2 will depend only on taxes and expenditure in period 2. Thus the districts' problem separates across periods. The choice of the period 2 variables $e_2,\tau_2$ is more complicated. These affect a district's objective through the period 1 prices, as given by equation (\ref{eq5.1}) below. However the affect that $e_2$ has on $p_2$, and thus on $p_1$, depends on the reaction of future governments to the future characteristics of district populations. Fortunately, since the districts' problem separates across periods, we can focus simply on the choice of $e_1$ and $\tau_1$. 

Consider the problem faced by the government of district $A$. I fix the strategy of district $B$ and suppress dependence on this in the notation that follows. Given the choice of $e^A_1$, the set of residents moving to district $A$ and the payment-price vectors for each house are uniquely determined. The population of new residents can be described by the lowest type that is indifferent between districts $A$ and $B$, denoted by $w_*(e^A_1)$. Let $m^A(q|e^A_1)$ be the PDV of the payment-price vector for house $q$. Then $p^A(q)$ is determined by
\begin{equation}\label{eq5.1}
    p_1^A(q|e_1,\tau_1) = - \tau_1(q) + \dfrac{1}{1+r}p_2^A(q|e_2, \tau_2) - m^A(q|e_1).
\end{equation}
By Lemma (\ref{lem5.7}), $p_2^A(q|e_2,\tau_2)$ is independent of the period 1 tax schedules. Fixing $e^A_2, \tau^A_2$ and the actions of the district $B$ government, the objective of district $A$ is 
\begin{equation*}
    \Pi^A(e_1, \tau_1) = \theta (s(e_1) + s(e^A_2)) + \int u(\tilde{w}^A_q + p^A(q|e_1,\tau_1) ) dQ^A,
\end{equation*}
where $p^A(q| e_1, \tau_1)$ is given by (\ref{eq5.1}). Only the expenditure levels of a given district affect the payoffs of the other; the specific tax schedules used to raise revenue have no cross-district impact, conditional on expenditure levels. The optimal tax schedule with which a district can raise its specified revenue is a function of the expenditure levels chosen by each of the two districts. Thus we can write each district's payoffs simply as a function of the profile of expenditure levels. Continuity of $u$ and $s$ guarantees that these payoffs are continuous. An equilibrium in district expenditure levels is an equilibrium in this game. Existence of such an equilibrium follows from standard results when it is possible to bound the space of feasible expenditure levels.

\begin{lemma}\label{lem:existence}
Assume $u$ and $s$ are continuously differentiable, and $\lim_{x \rightarrow \infty} s'(x) = 0$, $\lim_{x \rightarrow 0} u'(x) = \infty$. Then an equilibrium in district expenditure levels exists.
\end{lemma}

The important implication of Proposition \ref{prop2} is that districts will over-tax relative to what would be optimal given the equilibrium wealth distribution of new arrivals. Formally, let $\bar{e} = e^A_1$ be the equilibrium choice of expenditure by district $A$. Fix the gap between the expenditure levels of $A$ and $B$ at  $e^A_1 - e^B_1$, and consider the problem of district $A$ choosing its period 1 actions given its period 2 actions and the expenditure gap with $B$. This amounts to modifying the objective from $\Pi^A$ to
\begin{equation*}
    \tilde{\Pi}^A(e_1, \tau_1| \bar{e}) = \theta (s(e_1) + s(e^A_2)) + \int u(\tilde{w}^A_q + \tilde{p}^A(q|e_1,\tau_1,\bar{e} )) dQ^A
\end{equation*}
where 
\begin{equation*}
    \tilde{p}^A(q|e_1,\tau_1, e^A_1) = - \tau_1(q) + \dfrac{1}{1+r}p_2^A(q|e^A_2, \tau^A_2) - m^A(q|\bar{e}). 
\end{equation*}
The difference between $\tilde{p}^A$ and $p^A$ is that $m^A$ is fixed in the latter but varies in response to $e^A$ in the former. Let $\tilde{e}^A_1(\bar{e})$ be the solution to this problem. Given a choice of $e_1$, let $\tau^A_1(e_1), \tilde{\tau}^A_1(e_1|\bar{e})$ be the optimal tax schedules for collecting $e_1$ when the objectives are $\Pi^A$ and $\tilde{\Pi}^A(\cdot,\cdot|\bar{e})$ respectively. When the district takes into account the effect that higher school quality has on the PVDs of its homes it will have an additional incentive to increase expenditure, above what it would if the PVDs were fixed at equilibrium levels. Put differently, if they could be guaranteed their equilibrium set of new arrivals, districts would prefer to have lower expenditure.

\begin{proposition}\label{prop5.8}
$ e^A_1 \geq \tilde{e}^A_1(e_1^A)$, with equality if and only if $s(e^A_1) + \bar{q}^A \leq s(e^B) + \und{q}^B$. The same conclusion holds switching the labels $A$ and $B$. 
\end{proposition}

The incentive for districts to engage in excessive expenditure, in the sense of Proposition \ref{prop5.8}, comes from the increased property values enjoyed by current residents. This depends on the fact that residents own their homes. Renters would suffer from increased property values. Renters can be accommodated in the model, so long as some districts continue to have a preference for high-income populations. In practical terms this is a reasonable assumption. Wealthy districts with a high percentage of home owners will likely be more concerned with supporting property values than on maintaining affordable rental rates. I will return to the role of renters when discussing policy recommendations when there are more than two districts.

Abusing notation, write $\tilde{\Pi}(e|e^A_1) \equiv \tilde{\Pi}(e, \tilde{\tau}^A_1(e|e^A_1)|e^A_1)$. As noted in the proof of Proposition \ref{prop5.8}, $\tilde{\Pi}(\cdot|\bar{e})$ is differentiable in $e$. In fact, concavity of $u$ implies that $\tilde{\Pi}(\cdot|\bar{e})$ is concave as well. Then Proposition \ref{prop5.8} implies that both districts would benefit by reducing their expenditure levels. An \textit{expenditure cap} is an upper bound on the level of expenditure a district can choose. 

\begin{proposition}
Assume $e^j_1 > \tilde{e}^j_1$ for $j = A,B$. For any period 1 expenditure caps $\bar{e}^A_1 \in [\tilde{e}^A_1(e^A_1), e^A_1)$ and $\bar{e}^B_1 \in [\tilde{e}^B_1(e^B_1), e^B_1)$ such that $\bar{e}^A_1 - \bar{e}^B_1 = e^A_1 - e^B_1$, both districts are strictly better off, and the welfare of each is strictly decreasing in the cap level. 
\end{proposition}

\subsection{More than two districts}\label{sec:multiple}

With two districts expenditure caps benefit both districts, but also lower school quality everywhere. However with more than two districts the policy maker can achieve a surprising outcome using expenditure caps: all districts can benefit strictly from expenditure caps, and some can have strict improvements in school quality. The intuition is simple. It is still the case that when one district increases its expenditure it attracts a wealthier population of new arrivals. Moreover a version of Proposition \ref{prop2} continues to hold; this change in the incoming population leads to worse payment-price vectors, which translate to higher home prices in equilibrium. Thus districts will over tax. The novel feature with more than two districts is that each district will over tax relative to what it would if its expenditure gap with one of the remaining districts was fixed, holding fixed the expenditure level of the other district. Thus expenditure caps below the equilibrium level for only two districts improve both of the capped districts' welfare. The remaining district will have a richer population, and can appropriate some of the surplus by increasing expenditure. 

With more than two districts the location allocation must still be monotone, for the same reasons as before, and will take the form outlined in Proposition \ref{prop1}. The key extension is a generalization of Proposition \ref{prop2} to the case of more than two districts. That is, increasing expenditure in district $j$ while holding that of the other districts fixed leads to lower PVDs in district $j$ and higher PVDs in other districts. 

\begin{proposition}\label{prop4}
Consider expenditure profiles $\{e^j \}_{j=1}^N$ and $\{\hat{e}^j \}_{j=1}^N$ with $\hat{e}^1 > e^1$ and $\hat{e}^j = e^j$ for $j \neq 1$. Let $\hat{m}^j,m^j$ be the associated lump sum equivalents to the payment-price functions. Then
\begin{itemize}
    \item If $\und{q}^1+ s(e^1) \geq \bar{q}^j + s(e^j)$ for all $j \neq 1$ then $\hat{m}^j_q = m^j_q$ for all $q$ and $j \neq 1$, and $ m^1_q > \hat{m}^1_q$ for all $q$. 
    \item If $\und{q}^j+ s(e^j) \geq \bar{q}^1 + s(e^1)$ for all $j \neq 1$ then $\hat{m}^j_q > m^j_q$ for all $q$ and $j \neq 1$, and $ m^1_q = \hat{m}^1_q$ for all $q$. 
\end{itemize}
If neither of these conditions holds then 
\begin{itemize} 
    \item If $\und{q}^1 + s(\hat{e}^1) \geq \und{q}^j + s(\hat{e}^j)$ let $q^*$ satisfy $\und{q}^1 + s(\hat{e}^1) = q^* + s(\hat{e}^j)$. Then $m_q^1 > \hat{m}_q^1$ for all homes in district $1$. In district $j$ $\hat{m}_q^j > m_q^j$ for $q > q^*$ and $\hat{m}_q^j = m_q^j$ for $q \leq q^*$.
    \item If $\und{q}^1 + s(\hat{e}^1) < \und{q}^j + s(\hat{e}^j)$ let $\hat{q}^*$ satisfy $\hat{q}^* + s(\hat{e}^1) = \und{q}^j + s(\hat{e}^j)$. Then $m_q^1 > \hat{m}_q^1$ for all $q > q^*$ and $m_q^1 = \hat{m}_q^1$ for all $q \leq q^*$. In district $j$ $\hat{m}_q^j > m_q^j$ for all $q$.
\end{itemize}
Moreover, whenever $m^1_q > \hat{m}^1_q$ holds, 
\begin{equation*}
 \limsup_{\hat{e}^1 \searrow e^1} \dfrac{\hat{m}^1_q - m^1_q}{\hat{e}^1 - e^1 } < 0
\end{equation*}
and whenever $\hat{m}^j_q > m^j_q$, 
\begin{equation*}
 \limsup_{\hat{e}^1 \searrow e^1} \dfrac{\hat{m}^j_q - m^j_q}{\hat{e}^1 - e^1 } > 0.
\end{equation*}
\end{proposition}
The proof is essentially identical to that of Proposition \ref{prop2}.

Consider the game described in Section \ref{sec:district_game}, now with $N> 2$ districts. The relevant generalization of Proposition \ref{prop5.8} is the following. For a given district $k$ consider it's choice of expenditure level when the expenditure gap between $k$ and some set of districts is held fixed at the equilibrium level, holding fixed the expenditure levels of other districts not in this set. In this problem district $k$ will choose a lower level of expenditure than it would when the expenditure levels of all other districts are held fixed. 

To formalize this, fix a set of districts $Z \subset \{2, \dots, n \}$. Let $e^j_1$ be the equilibrium first period expenditure of district $j$. Consider the problem of district $1$ choosing it's expenditure level and tax schedule when the expenditure gap between it and all districts in $Z$ is held fixed, fixing the expenditure level of districts not in $Z$ at the equilibrium level. I refer to this as the ``fixed gap'' problem. Let $\tilde{\Pi}(e_1, \tau_1 )$ be it's objective in this case, and as before define $\tilde{\tau}(e_1)$ be the optimal period 1 tax schedule for collecting revenue $e_1$ and let $\tilde{e}^1_1$ be the optimal expenditure level. 

\begin{proposition}\label{lem6.9}
$ e^1_1 \geq \tilde{e}^1_1$, with equality if and only if $s(e^1_1) + \bar{q}^1 \leq s(e^j) + \und{q}^j$ for all $j \in Z$.
\end{proposition}

Given concavity of the fixed gap objective we have the following immediate implication.

\begin{proposition}\label{prop5}
Let $Z \subset \{1, \dots, N \} $ be a set of districts such that $e^j_1 > \tilde{e}^j_1$ for all $j \in Z$. For any first period expenditure caps $\bar{e}^j_1 \in [\tilde{e}^j_1, e^j_1)$ for all $j \in Z$ with $e^j_1 - \bar{e}^j_1 = e^k_1 - \bar{e}^k_1$ for all $j,k \in Z$, and fixing the expenditure of districts $j \not\in Z$ at $e^j_1$, all districts are weakly better off. All districts in $Z$ are strictly better off, and any district $k \not\in Z$ is strictly better off if there exists $j\in Z$ such that $\bar{q}^j + s(e^j_1) > \und{q}^k + s(e^k_1)$. The welfare of each district is decreasing in the cap level. 
\end{proposition}

Proposition \ref{prop5} gives the central policy recommendation of the paper. The central government can provide tax relief to a set of wealthier districts while reducing their excessive expenditure. This results in less geographic income inequality and improves the tax base in the remaining districts, leading to higher home values in these districts. These remaining districts can then increase their expenditure slightly without making the capped districts worse off.

Rather than imposing a hard cap, the policy can be formulated as a fee that must be paid by districts that engage in excessive expenditure, as is done in some U.S. states, e.g. Vermont. Revenue gathered from this fee can be transferred directly to the districts that need it most. This type of ``tax-on-tax'' policy will be discussed further below, as it may be necessary for generating Pareto improvements when there are some districts in which most residents rent rather than own their homes.   

Expenditure caps are in place in many states. An important take away from the model is that expenditure caps should be used selectively to correct geographic income imbalances. The analysis can also be seen as providing a novel justification for tax caps. Critics of such caps may argue that localities have a right to tax themselves to support high expenditure if they so choose. Such arguments, however, ignore the fiscal externalities of high expenditure. 

As noted above, renters reduce the incentives of local governments to compete for high-income residents and raise property values. As long as at least two districts engage in this type of competition however, the policy recommendation holds; these districts are the ones that can be capped. The central planner may also be concerned about rental rates. Reducing geographic income disparity inequality in school qualities will necessarily result in higher property values in the poorest districts. The central planner will need to use other tools to directly address higher rental rates.

\subsection{Importance of heterogeneous home quality}\label{sec:hetero_home}

Heterogenetiy of home qualities is a salient real-world feature. Moreover, models in which homes are undifferentiated may miss the important forces identified here. This may explain why the upward pressure on expenditure caused by competition for richer residents has not greater attention in the literature.  

Consider the two district model as described above, but in which all homes are identical. The within-district equilibrium will be similar: if $e^A > e^B$ then the richest half of the new arrivals will go to district $A$ and the rest to district $B$. The home price gap between the two districts will be determined by the indifference condition of the marginal household; the median-wealth individual should be indifferent between the two districts. Assume, as before, that there is a non-market housing option available to the poorest individuals. Then home prices in district $B$ will be pinned down by the indifference of the poorest type between the non-market and market housing in district $B$. Then home prices in district $A$ are determined by the indifference condition of the median type. 

An important difference between this model and that which is the focus of this paper is that the expenditure level in district $A$ will not affect the home price in district $B$, so long as $e^A > e^B$. This observation generalizes to any number of districts: the specific expenditure level in higher-expenditure districts does not matter for those with lower expenditure. As a result, there cannot be Pareto improvements through expenditure caps. 

\subsection{Future periods}

Thus far I have discussed policies that affect only the current period tax schedules. I show (Lemma \ref{lem5.7}) that changes to the first period schedules will not affect the actions of future district governments or new arrivals. In general however, the policy maker will be interested in making changes over many periods. Nonetheless, the political constraint on which policies can be enacted is given by the preferences of current residents and district governments. The question then is what long-term policies can be enacted without making any current districts worse off. Given that a policy has been identified that is strictly (under some conditions) Pareto-improving in the current period, continuity of the district objectives guarantees that there are Pareto improving policies which affect expenditure in all periods. In particular, the central government can impose a sequence of expenditure caps on a set of districts $Z$ in all periods. This makes districts outside of $Z$ better off. It may be necessary to impose caps districts outside of $Z$ as well to ensure that their increase in expenditure does not erase the tax-relief gains for districts in $Z$. These caps will still be above the level of expenditure that districts outside of $Z$ would have chosen in the absence of any policy intervention.

\subsection{Renters}\label{sec:renters}

Suppose that in each district some exogenous fraction of the homes are only available to be rented. Let $Q^j = O^j + R^j$, where $O^j$ be the conditional CDF of home quality for owner-occupied homes in district $j$, and $R^j$ the conditional CDF for rental homes. The key difference between renters and owners is that the former, once they have made their location choice, are hurt by increases in home prices (as these translate to higher rents), while the later benefits from increases in home prices. This is most easily illustrated in the current model if we assume that renters will remain in their chosen district for two periods. Renters cannot sign long-term contracts; the rate they pay in each period is determined by the current housing market equilibrium. Owners, as before, sell their home after one period.\footnote{This asymmetry is not important. The key is that renters may be hurt by increases in home values after they have made their location choice. We could assume instead, for example, that all individuals live for three periods and reside in their chosen district for the first two.}

As noted in Section \ref{sec:housing_mkt}, the number of periods does not affect the conclusions of that section. A renter of a quality $q$ home in district $j$ pays $f^j_t(q) + \tau^j_t(q)$ in periods $t = 1,2$, where $f^j_t(q)$ is the rental rate net of tax. The renter, of course, receives no value from the sale of the home when they leave. Like buyers, renters are able to borrow and save to smooth consumption over the two periods. Rental contracts last for only one period; the renter cannot lock in their second period rate.  

It is easy to see that the presence of rental homes does not affect the monotonicity of the equilibrium location allocation identified in Lemma \ref{lem0.1}: by exactly the same argument as before, if a low-income individual prefers a high quality location to a low quality one, regardless of whether these are rental or owner-occupied, then so will a high-income individual. Recall that we call an equilibrium monotone if the allocation of location qualities is monotone in individuals' types. 

\begin{lemma}\label{lem:retal_monotonicity}
Any equilibrium is monotone, regardless of the composition of rented v.s. owner-occupied properties.
\end{lemma}

Proposition \ref{prop1} describes the unique monotone location allocation which must necessarily arrive in equilibrium. Note, moreover, that the results of Section \ref{sec:housing_mkt} refer only to the PDV of the equilibrium payment-price vectors. Thus all results of this section, in particular Proposition \ref{prop2}, continue to hold regardless of the composition of rented and owner-occupied properties. To put this another way, in equilibrium the present discounted value of consumption can depend only on location quality, not on whether an individual chooses to rent or buy. This means that for a renter arriving in period 1
\begin{equation}\label{eq:rental_rate}
    f^j_1(q) + \tau^j_1(q) + \frac{1}{1+r}\left( f^j_2(q) + \tau^j_2(q)\right) = - m^j(q)
\end{equation}
where $m^j(q)$ is, as before, the equilibrium PDV when there are no renters. 

We can now examine the effect that the division of the housing market between rental and owner-occupied houses has on district government actions. The game played by districts will be the same as that considered in Section \ref{sec:district_game}. The district objective is to maximize school quality and the welfare of current residents, i.e. those who arrived in the previous period. With renters, this objective is given by 
\begin{align*}\label{eq:renter_objective}
    \Pi^j(e^j_1, e^j_2,\tau^j_1, \tau^j_2) &= \theta \left(s(e^j_1) + s(e^j_2)\right) \\
    &+ \int u(\tilde{w}_O^j(q) + p^j(q)) dO^j +
    \int u(\tilde{w}_R^j(q) - f^j_1 - \tau^j_1(q)) dR^j,
\end{align*}
where $p^j(q)$ depends on $e^j_1, e^j_2, \tau^j_1$, and $\tau^j_2$, as given by equation (\ref{eq5.1}). Here $\tilde{w}^j_O$ and $\tilde{w}^j_R$ are the sum of income and net savings for owners and renters respectively. Although renters and owners in a given location have the same PDV of consumption, their savings behavior will generally differ since an owner-occupied home is also an instrument for saving. 

Lemma \ref{lem5.7} continues to hold, so the districts' problem of choosing tax schedules remains separable across periods. Thus, as before, we can focus on the problem of choosing $e^j_1$ and $\tau^j_1$. The determination of home prices is as before. For renters, this means that $f^j_2(q)$ is independent of $e^j_1$ and $\tau^j_1$. The period 1 rental rate is then determined by $e^j_1$ and $\tau^j_1$ through (\ref{eq:rental_rate}), where $f^j_2(q)$ is fixed. A reduction in $m^j(q)$ induced by increases in the relative quality of district $j$ schools leads to higher rental rates. On the other hand, such changes benefit owners through higher home prices. When the proportion of the district $j$ housing stock that is made up of rental properties is high, Proposition \ref{prop4} will therefore imply that district $j$'s equilibrium spending, $e^j_1$, is lower than what it would spend in the fixed-gap problem $\tilde{e}^j_1$ (i.e. the reverse of the conclusion of Proposition \ref{lem6.9}). If this is true for all districts then, by the same argument as Proposition \ref{prop5}, all districts would benefit from an expenditure floor. 

It is interesting to consider what policies would be Pareto improving when there is a set of districts $Z$ that have a low proportion of renters (so $e^k_1 > \tilde{e}^k_1$ for $k\in Z$), and the rest have a high proportion of renters (so $e^j_1 < \tilde{e}^j_1$ for $j \not\in Z$). In this case an expenditure cap on districts in $Z$ would benefit those districts, but hurt the remaining districts. An exception to this would be if the support of location qualities for districts in $Z$ did not overlap with those outside of $Z$. In this case a cap for districts in $Z$, provided it was not too low, would not lead to any change in the location allocation. In this case a ``tax-on-tax'' policy for districts in $Z$ would be necessary to transfer some of the surplus generated by the tax relief provided by a cap on these districts to districts outside of $Z$.

The preceding analysis of housing market equilibrium with rental properties also sheds some light on observed patterns of home ownership, location choice, and income. The rate of home ownership in the U.S. is lower for poor households than for rich. Moreover, there is a high degree of residential segregation based on income, with poorer households concentrated in neighborhoods with a high fraction of renters. Existing research on this phenomenon highlights the role of credit constraints faced by poorer households, along with historical practices of racial discrimination in the housing market. Poor households have difficulty saving for a down payment, and face higher interest rates on their mortgages due to greater income volatility. (\cite{carasso2005improving}). While credit constraints and discrimination are certainly important constraints on the ability of poor households to buy a home, the preceding analysis shows that the observed pattern of home ownership may arise even without these factors (households here are not credit constrained, and there is no explicit discrimination). In this context, poor households concentrate in districts with a large number of rental properties for two reasons. First, as discussed above, the aversion to increasing home prices among renters leads such districts to under-invest in the local public good. Given this under investment, monotonicity of the housing market equilibrium implies that the income distribution of new arrivals will be lower in such districts. Thus patterns of local public goods provision reinforce credit constraints and discrimination. 

This analysis also suggests a role for housing policy in alleviating education inequalities. Policies aimed at improving homeownership among lower-income households should should incentivize these households to purchase homes \textit{within} the districts in which they currently live. This will help such households enjoy the property-value benefits of improving school quality. 

\subsection{Extensions}
One interesting extension would be to study the dynamic  implications of districts having different compositions of rented v.s. owner-occupied homes. As districts with a high proportion of renters have a lower incentive to encourage high-income new arrivals, and thus set lower taxes, they will further reduce their ability to raise revenue in future periods. This should create a self-reinforcing cycle of polarization, with poor renters concentrating in districts with low public goods provision. One could also consider multi-dimensional public goods. Presumably the relative salience of some features for potential home buyers would distort incentives for public goods provision. 

\section{Empirical evidence}\label{section:empirical}

The fact that equilibrium expenditure is too high in all districts, which implies that tax caps can be Pareto improving and increase school quality in some districts, is driven by the sorting of higher income individuals into districts with higher expenditure. Districts with higher expenditure, and thus better schools, attract the richer individuals of the incoming generation. However the distribution of income across districts is essentially a zero sum game. Districts that increase expenditure and attract richer new arrivals impose a negative externality on neighboring districts, inducing lower wealth levels and property values in these districts. 

Using data from Massachusetts, I seek to identify whether income sorting of this type occurs. Massachusetts' Proposition 2 1/2 imposes caps on both the level and the rate of change of taxation, but municipalities can hold referenda to increase the cap. I use close referenda to identify plausibly exogenous changes in expenditure. I use these changes to identify the impact of increased expenditure on the income distribution and property wealth within a district and in neighboring municipalities. The model has been discussed in the context of school funding, but it applies more generally to any problem of local public goods provision. In the empirical analysis I consider all such investments, and do not differentiate between different types of local public goods.  

\subsection{Institutional setting}

Property taxes in Massachusetts are levied exclusively by municipalities, unlike some states in which there are multiple taxing bodies that may levy taxes on a given property. In 1981, in an effort to limit local property taxation, the Massachusetts state legislature passed Proposition 2 1/2, which took effect in 1983. Proposition 2 1/2 imposes two caps on the amount that a municipality can levy. First, a municipality cannot levy more than 2.5 percent of the assessed property base (including real and personal property). This limit is know as the \textit{levy ceiling}. Additionally, the law controls the rate at which the levy can increase. The maximum allowable levy in a given year is called the \textit{levy limit}. The levy limit cannot exceed the levy ceiling. The levy limit increases by 2.5 percent each year.\footnote{The levy limit may also increase as a result of new growth in the tax base that is \textit{not} the result of reevaluation.}

Both the levy limit and levy ceiling can be increased through a local referendum. An increase in the levy limit is known as an \textit{override}. This action can be taken only if a community is below its levy ceiling. An override has a permanent effect on the levy limit, as the higher levy limit is used to calculate that of future years. The levy ceiling cannot be permanently changed, but it can be temporarily increased (known as an exclusion) through a referendum for the purposes of raising funds for debt service costs (debt exclusion) or for capital expenditures (capital outlay expenditure exclusion). 

\subsection{The data}

I obtained data on Proposition 2 1/2 referenda, dating to the introduction of the regulation in 1983, from Massachusetts' Division of Local Services (DLS). Municipalities are required to report on referenda that have passed, as this is necessary to ensure that they are complying with the law. Districts also report referenda that do not pass. Unfortunately reporting of failed referenda is not required by law. From the frequency of reported losses in the dataset it seems that around half of all losses are not reported. In a given year, a municipality may hold multiple referenda. The only discernible pattern in the reporting of losses is that in years in which there is also a win, the reporting of which is mandatory, municipalities are more likely to also report losses, compared to years in which there was no win. In order to address potential selection issues associated with this pattern of reporting, I restrict attention to municipality-year observations in which there were either only losses or only wins. Within this sample there does not seem to be any pattern in which municipalities report losses. Since reporting losses has no consequences for municipalities, it seems reasonable to assume that the decision to report a loss is orthogonal to all other variables of interest. I maintain this assumption throughout. 

The data include both the year of the vote and the year in which the measure was to take effect. Most successful Overrides and Capital Exclusions take effect the year after the vote is held. Overrides, however, have a cumulative effect on the levy limit over time. Debt Exclusions on the other hand are often in effect for multiple years, and may not take effect in the year immediately following the vote. To accommodate both the cumulative effect of overrides and the delayed effect of Debt Exclusions I record referenda according to their vote date. Thus we would generally expect a referendum in year $t$ to begin to effect the levy in year $t+1$. 

In order to identify changes in income I use municipal level income data, based on tax returns, provided by the Massachusetts Department of Revenue (DOR). I also incorporate reports from the DLS and DOR on residential property values, home values, single-family home values, tax levies, and average single family tax bills by municipality. 

The key variables of interest are the changes in tax levies and income per capita in each municipality, and in neighboring municipalities. I use GIS data to determine the neighbors of each municipality, determined by geographical contiguity. Approximately half of Massachusetts municipalities belong to some form of multi-town school district. Moreover, around 30\% of referenda in the data concern school expenditure. Since the interest is on expenditure externalities, I refine the set of neighbors for each district by considering only those that do not belong to the same school district. I also exclude Boston from the analysis. Boston has had no close referenda, and given it's size relative to neighboring municipalities, the aggregate measure of property values is unlikely to be much impacted by the expenditure of any of it's neighbors.

\subsection{Empirical strategy}

The empirical strategy is to use sharp and fuzzy regression discontinuity designs (RDD), where the running variable is the percent of ``Yes'' votes in a referendum. I relate the discontinuity in the passage of a referendum at the 50\% yes vote margin, and the induced discontinuity in property tax levels, to changes in the outcomes of interest. 

Municipalities frequently hold multiple referenda per year. For municipality-year observations in which there are multiple referenda I use the average vote margin as the running variable. For the analysis presented here, I aggregate these observations, and simply tally whether or not a municipality had a close win or a close loss. As discussed above, I also restrict attention to observations for which there are either only losses or only wins. For this sample the average vote margin is the correct running variable for determining whether a measure passed or failed.  

\subsection{Results}

I first examine the effect of referendum outcomes on tax levies. I focus here on the average tax bill for single family households; similar patterns hold for the aggregate level of residential property taxation. Let $Margin_{i,t}$ be the fraction of ``Yes'' votes in municipality $i$ and year $t$ (averaged across referenda if there were multiple votes), and $W_{i,t} := \mathbbm{1}_{Margin_{i,t} \geq 0}$ indicate whether or not the referendum passed.\footnote{Recall that I am restricting attention to observations in which either all votes were wins or all were losses.} Define 
\begin{align*}
    M_{i,t} := &\Big(1, Margin_{i,t}, Margin_{i,t}^2, Margin_{i,t}^3,\\
    &W_{i,t}\times Margin_{i,t}, W_{i,t}\times  Margin_{i,t}^2, W_{i,t}\times Margin_{i,t}^3\Big)'
\end{align*}
which is a vector of polynomial terms and interactions, included to allow for functional form flexibility.\footnote{The results continue to hold using local estimation techniques. Scatter plots showing the relationship between vote margin and levy at different lags are included in Appendix \ref{app: empirical results}.}

Let $AvgTax_{i,t}$ be the average tax bill for a single family home in municipality $i$ and period $t$. Let $\tilde{\Delta}AvgTax^k_{i,t} := (AvgTax_{i,t + k} - AvgTax_{i,t})/AvgTax_{i,t}$ be the growth rate in the average tax between year $t$ and $t+k$. I estimate a set of sharp RDDs given by 
\begin{equation}\label{eq:levypercentchange}
    \tilde{\Delta}AvgTax^k_{i,t} = \gamma'M_{i,t} + \kappa \cdot W_{i,t} + \varepsilon_{i,t}.
\end{equation}
Figure \ref{fig:levypercentchange} displays the coefficient on $W_{i,t}$ for different lags $k$. As expected, passage of a referendum is associated with higher growth, measured as percent change, with a larger effect for longer lags.
\begin{figure}[ht]
    \centering
    \caption{Discontinuity in levy growth}
    \includegraphics[scale = .25]{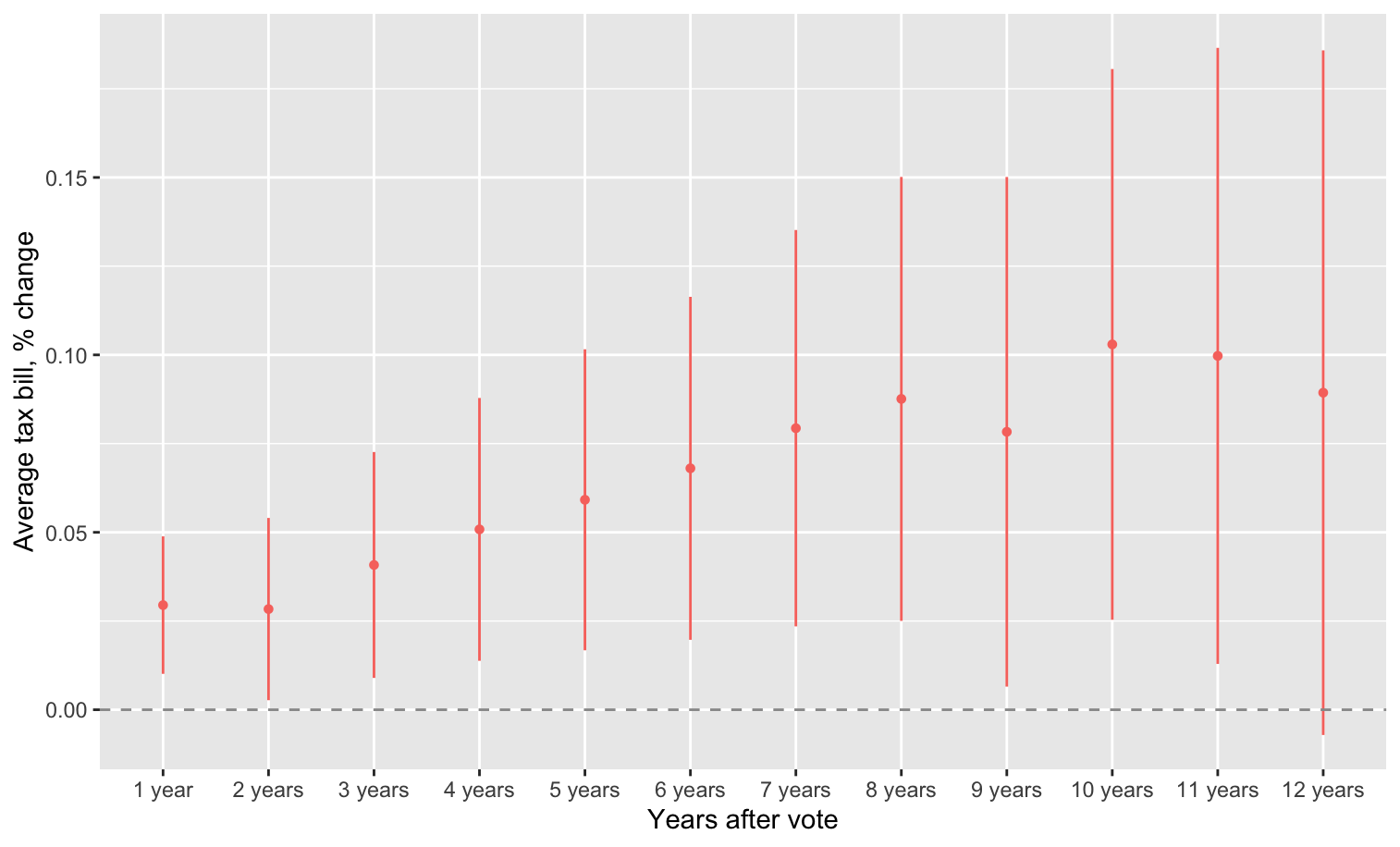}
    \label{fig:levypercentchange}
\end{figure}
I also estimate the discontinuity in tax levels associated with passage of a referendum. The regression equation is 
\begin{equation}\label{levylevel}
    AvgTax_{i,t+k} = \gamma'M_{i,t} + \kappa \cdot W_{i,t} + \varepsilon_{i,t}.
\end{equation}
The coefficients on $W_{i,t}$ for different lags $k$ are displayed in Figure \ref{fig:levylevel}.
\begin{figure}[ht]
    \centering
    \caption{Discontinuity in levy level}
    \includegraphics[scale = .25]{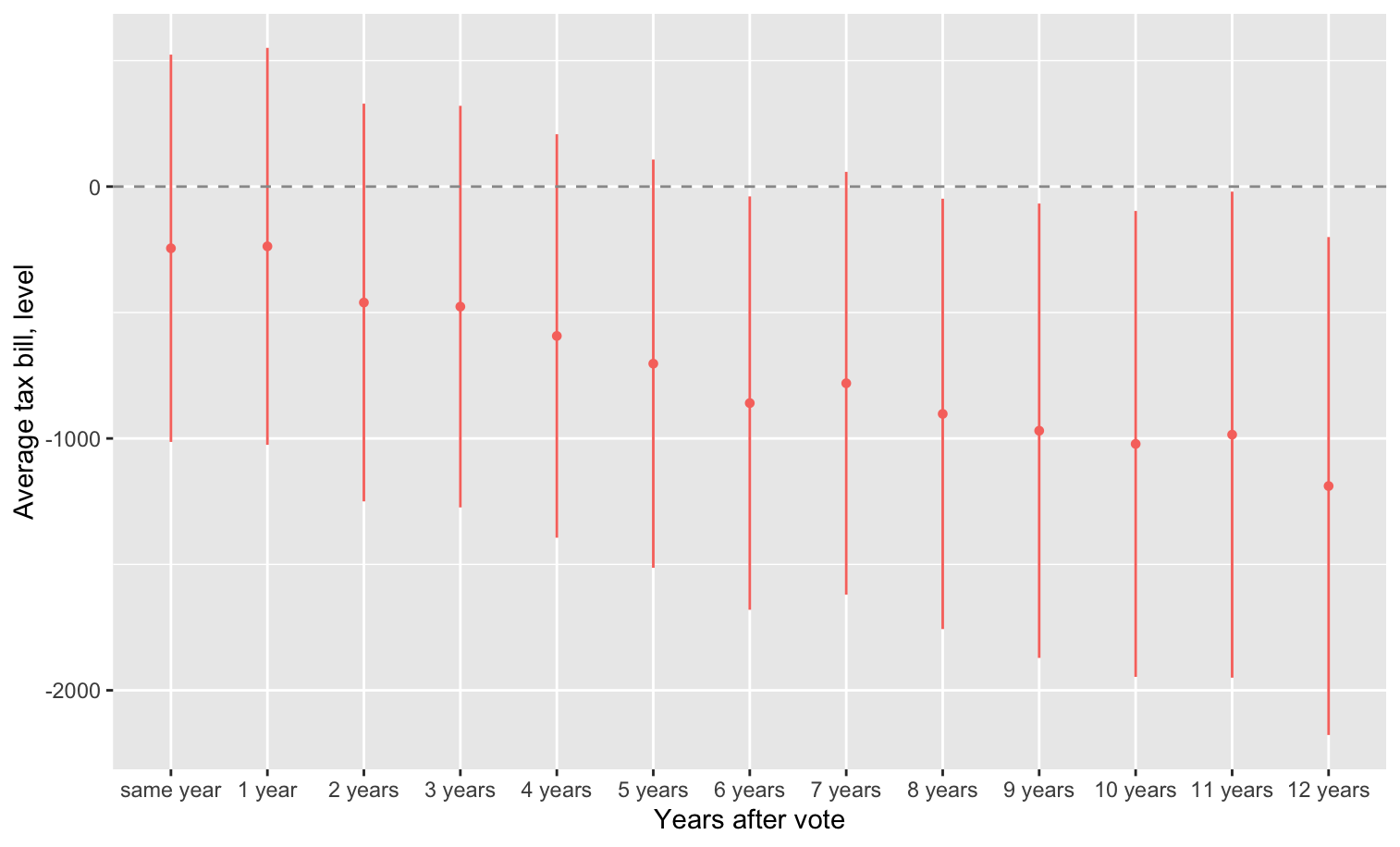}
    \label{fig:levylevel}
\end{figure}
Surprisingly, in the long run passage of an exclusion or override measure appears to have a negative effect on the level of  the tax levies. This pattern reflects, to some extent, the fact that income levels appear to be slightly lower to begin with in municipalities that pass a referendum. While this difference is not significant, it may compound over time; despite a higher growth rate, the levy in the lower income districts increases less than that of the higher income districts. Another explanation for this pattern could be that passage of an override or exclusion makes it less likely that another such measure will pass later on.

\subsubsection{Income sorting}
Looking directly at income data, I find some evidence that higher income households sort into districts with higher tax levies. However the relationship between levies and incomes is significant only in levels; I find no evidence that growth in the levy is associated with growth in incomes per capita. The model of interest is
\begin{equation}\label{eq:incomePC}
    IncomePC_{i,t} = \beta_0 + \beta_1 AvgTax_{i,t} + \varepsilon_{i,t}
\end{equation}
where $IncomePC_{i,t}$ is the income per capita in municipality $i$ and year $t$. The coefficient of interest, $\beta_1$, is identified using a fuzzy RDD exploiting the discontinuity in the tax bill induced by passage of a referendum, estimated by:
\begin{equation}
    AvgTax_{i,t} = \gamma' M_{i,t-j} + \kappa \cdot W_{i,t-j} +\sigma_{i,t}, 
\end{equation}
where $Margin_{i,t}$ is the fraction of ``Yes'' votes in municipality $i$ and year $t$, and $W_{i,t} := \mathbbm{1}_{Margin_{i,t} \geq 0}$ indicates whether or not the referendum passed. Figure \ref{fig:levylevel} shows that there is a significant discontinuity in the levy only after a few years have elapsed from the data of the vote. I estimate the model using different lags $j = 5, \dots, 12$ in the first stage to identify $\beta_1$, and report the results for each in Figure \ref{fig:incomePC}. It is worth emphasising that this figure reports different estimates of the same coefficient: the contemporaneous effect of increasing the tax levy on income per capita in a municipality.\footnote{Estimates of the effect of the levy on incomes in the subsequent year are similar.} The difference is in the identification strategy. The estimate of $\beta_1$ is interpreted as the local average treatment effect (LATE). 
\begin{figure}[ht]
    \centering
    \caption{Effect of tax level on income per capita}
    \includegraphics[scale = .25]{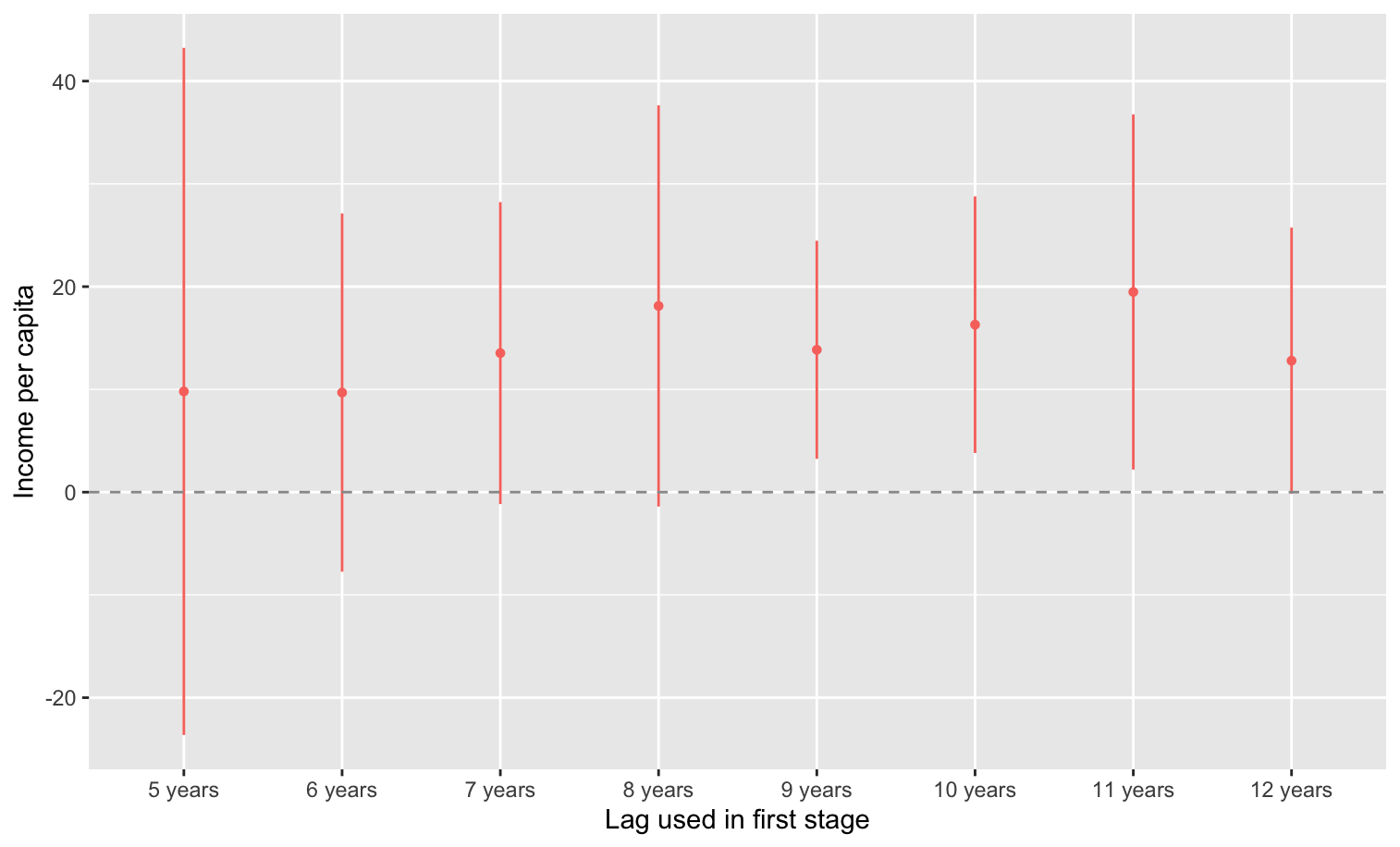}
    \label{fig:incomePC}
\end{figure}
All identification strategies give an estimate of the increase in income per capita resulting from a \$1 increase in the per family tax levy of between \$10 and \$20.

\subsubsection{Home values}

Further evidence of income sorting in response to tax increases comes from changes in home values. A benefit of looking at home values for evidence of income sorting, rather than at incomes per capita, is that the former reflects \textit{marginal} changes to the population, home prices being driven largely by characteristics of the new arrivals, whereas the latter is a population mean. We should thus expect home prices to be more sensitive to income sorting. In the model of Section \ref{sec:housing_mkt} the effect of increased taxation on home values in a municipality is ambiguous. On one hand, higher quality public goods increase demand for housing. However the higher taxes are also capitalized into home prices, leading to lower prices. The net effect depends on the relative size of these effects. More generally, it is unlikely that both home values and tax levies increase without some degree of income sorting. Presumably the increased levy translates to some degree into desirable local public goods in the given municipality. It would be surprising if both the cost of living and the quality of local public goods increased without higher incomes of new arrivals.  

I find strong evidence that an increase in taxation cause increases in home values. This holds even when I control for differences in initial conditions between municipalities that pass referenda and those that do not by looking an growth rates, rather than levels. I focus here on assessed single family home values. The results are similar if I look at assessed values for all residential property, or use Zillow's home value index rather than assessed values. 

Let $HomeValue_{i,t}$ be the average single family home value, and $\tilde{\Delta}HomeValue^k_{i,t}$ be the percent change from year $t$ to year $t+k$. As a first step, I estimate the following sharp RDD
\begin{equation}\label{eq:sharpRDD_ownvalue}
    \tilde{\Delta}HomeValue^k_{i,t} = \gamma' M_{i,t} + \kappa \cdot W_{i,t} + \varepsilon_{i,t}
\end{equation}
for $k = 1, \dots , 12$. The resulting estimates for $\kappa$ are summarized in Figure \ref{fig:sharpRDD_ownvalue}.
\begin{figure}[ht]
    \centering
    \caption{Discontinuity in home value growth}
    \includegraphics[scale = .25]{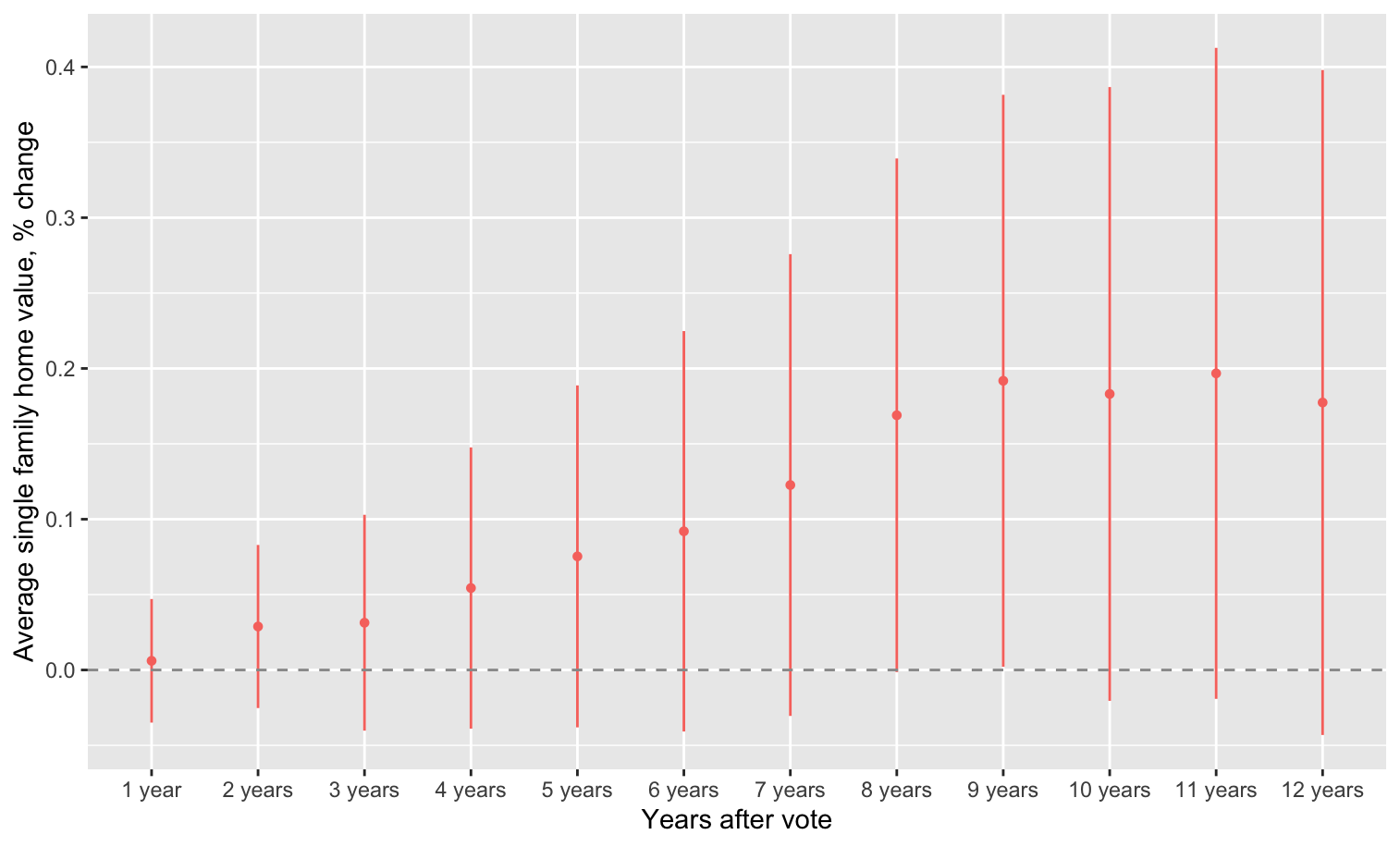}
    \label{fig:sharpRDD_ownvalue}
\end{figure}
The estimates obtained using a local linear estimator, with the Imbens-Kalyanaraman optimal bandwidth selection, are almost identical, and are summarized in Table \ref{tab:sharpRDD_ownvalue}.

\begin{table}[ht]
    \centering
    \begin{tabular}{c c c c c c c c c }
    \hline
       Lag & 5 yrs & 6 yrs& 7 yrs& 8 yrs& 9 yrs& 10 yrs& 11 yrs& 12 yrs \\
       \hline 
       \hline
       Estimate  & 0.0766 & 0.0915 & 0.1239* & 0.1569* & 0.1801* & 0.1749* & 0.1603 & 0.1558  \\
       p-value & 0.05495 & 0.06364 & 0.0245 & 0.0170 & 0.0123 & 0.0324 & 0.1262 & 0.0968 \\
    \hline
    \end{tabular}
    \caption{Discontinuity in home value growth}
    \label{tab:sharpRDD_ownvalue}
\end{table}

Do estimate the effect of the levy on the levy on home values I again use a fuzzy RDD. For contemporaneous effects, the model of interest is
\begin{equation*}
    \tilde{\Delta} HomeValue^k_{i,t} = \beta_0 + \beta_1 \cdot \tilde{\Delta}AvgTax^k_{i,t} + \varepsilon_{i,t}
\end{equation*}
where $\tilde{\Delta}AvgTax^k_{i,t}$ is the growth rate in the average tax bill for single family homes. The estimated LATE for different lags are summarized in Figure \ref{fig:fuzzyRDD_avgfamval}. Estimates using local linear regression, with Imbens-Kalyanaraman optimal bandwidth, are reported in Table \ref{tab:fuzzyRDD_avgfamval}.   
\begin{figure}[ht]
    \centering
    \caption{Effect of tax growth on home value growth}
    \includegraphics[scale = .25]{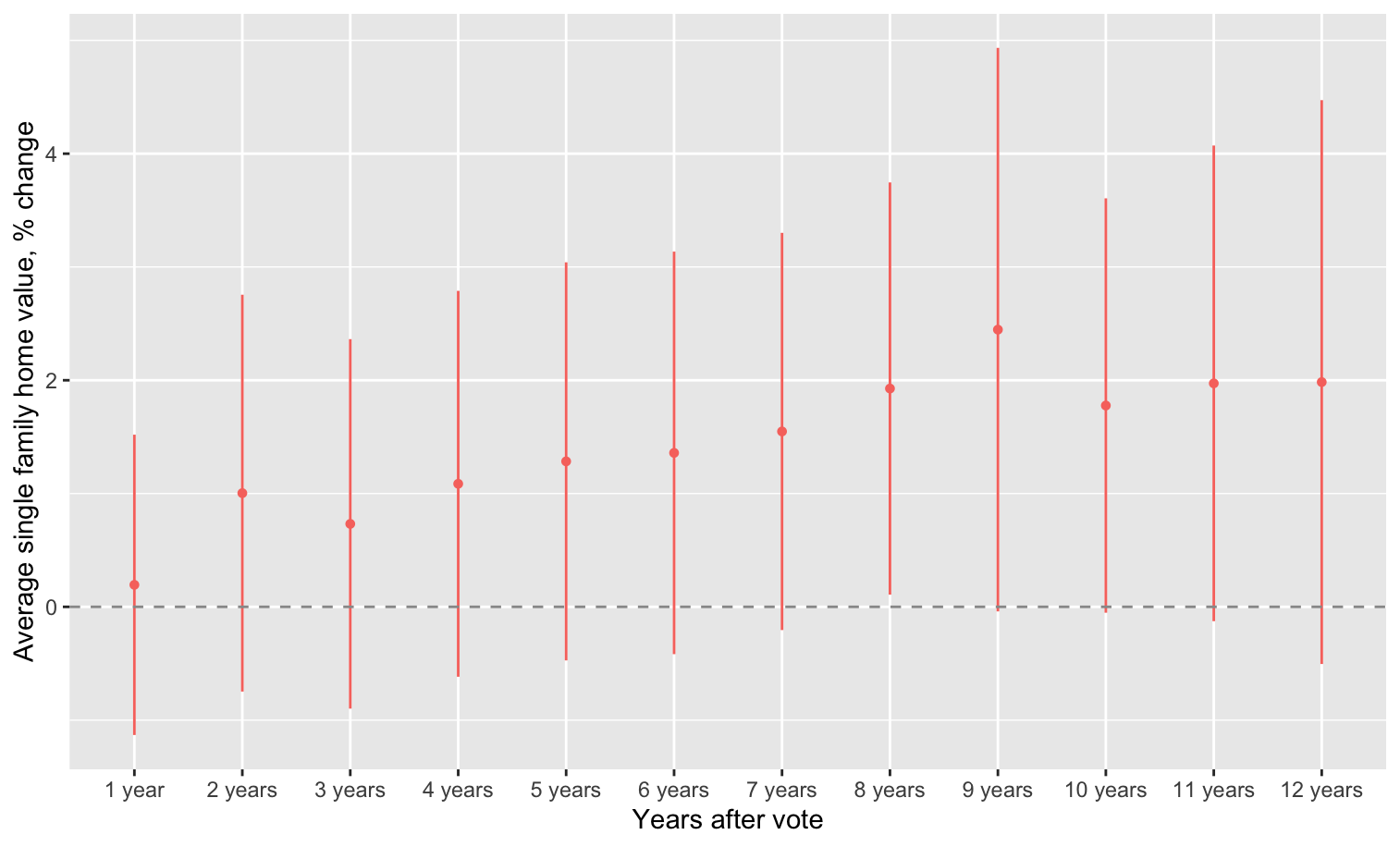}
    \label{fig:fuzzyRDD_avgfamval}
\end{figure}

\begin{table}[ht]
    \centering
    \begin{tabular}{c c c c c c c c c }
    \hline
       Lag & 5 yrs & 6 yrs& 7 yrs& 8 yrs& 9 yrs& 10 yrs& 11 yrs& 12 yrs \\
       \hline 
       \hline
       Estimate  & 1.351 & 1.583 & 1.765* & 2.246** & 3.047* & 2.392* & 2.5543 & 2.744  \\
       p-value & 0.0489 & 0.0121 & 0.0260 & 0.0061 & 0.0202 & 0.0152 & 0.0793 & 0.0212 \\
    \hline
    \end{tabular}
    \caption{Effect of tax growth on home value growth}
    \label{tab:fuzzyRDD_avgfamval}
\end{table}

In addition to the contemporaneous effect of levy growth, we can study the delayed effect of levy growth on home values in later years. Consider the model 
\begin{equation*}
    \tilde{\Delta} HomeValue^k_{i,t} = \beta_0 + \beta_1 \cdot \tilde{\Delta}AvgTax^{1}_{i,t} + \varepsilon_{i,t}.
\end{equation*}
In this case the parameter of interest is the LATE of levy growth in the first year following a vote on the growth in home values over the next $k$ years. Again, I estimate $\beta_1$ at for different time horizons using a fuzzy RDD. The results using a polynomial specification are summarized in Figure \ref{fig:fuzzyRDD_avgfamvalcumulative}. These are almost identical to those obtained from local linear regression using the Imbens-Kalyanaraman bandwidth, reported in Table \ref{tab:fuzzyRDD_avgfamvalcumulative}.

\begin{figure}[ht]
    \centering
    \caption{Effect of tax growth on lagged home value growth}
    \includegraphics[scale = .25]{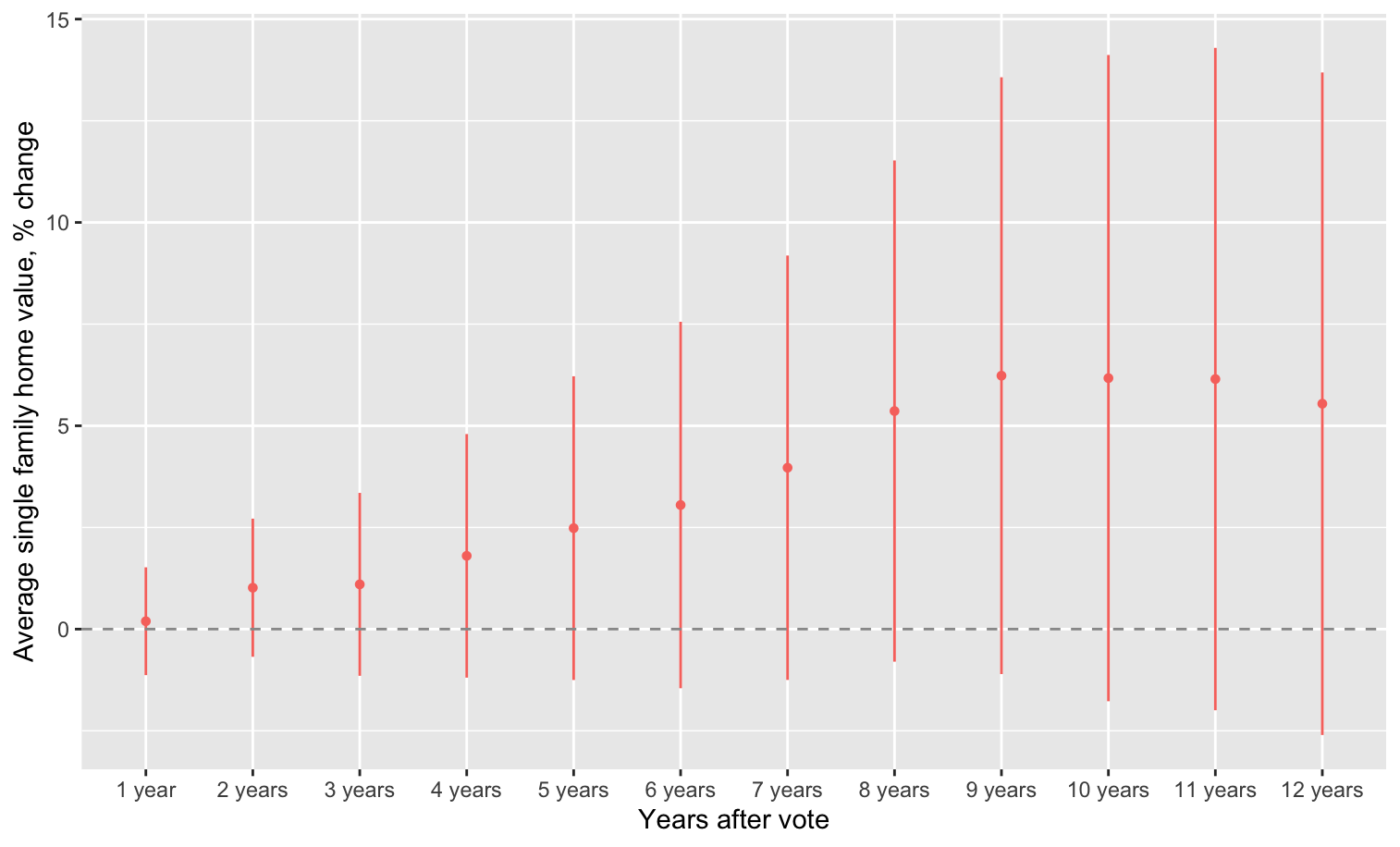}
    \label{fig:fuzzyRDD_avgfamvalcumulative}
\end{figure}

\begin{table}[ht]
    \centering
    \begin{tabular}{c c c c c c c c c }
    \hline
       Lag & 5 yrs & 6 yrs& 7 yrs& 8 yrs& 9 yrs& 10 yrs& 11 yrs& 12 yrs \\
       \hline 
       \hline
       Estimate  & 2.685 & 3.297 & 4.646* & 5.675* & 6.601* & 6.840* & 5.613 & 5.629  \\
       p-value & .0545 & 0.0777 & 0.0234 & 0.0214 & 0.0253 & 0.0410 & 0.1447 & 0.1146 \\
    \hline
    \end{tabular}
    \caption{Effect of tax growth on lagged home value growth}
    \label{tab:fuzzyRDD_avgfamvalcumulative}
\end{table}

\subsection{Neighboring municipalities}

In the model of Section \ref{sec:housing_mkt}, a district that increases its expenditure imposes a negative externality on other districts by attracting the higher income new arrivals. I find however that the effect of increases in municipal levies on neighboring district home values is positive, and of similar magnitude the change in the home values within the municipality. Figure \ref{fig:fuzzyRDD_nbravgfamvalcumulative} displays fuzzy RDD estimates of the effect of tax growth in the first year following a vote on the average growth rate in single family home values in neighboring municipalities (analogous to the results in Figure \ref{fig:fuzzyRDD_avgfamvalcumulative}). 

\begin{figure}[ht]
    \centering
    \caption{Effect of tax growth on lagged home value growth in neighboring municipalities}
    \includegraphics[scale = .25]{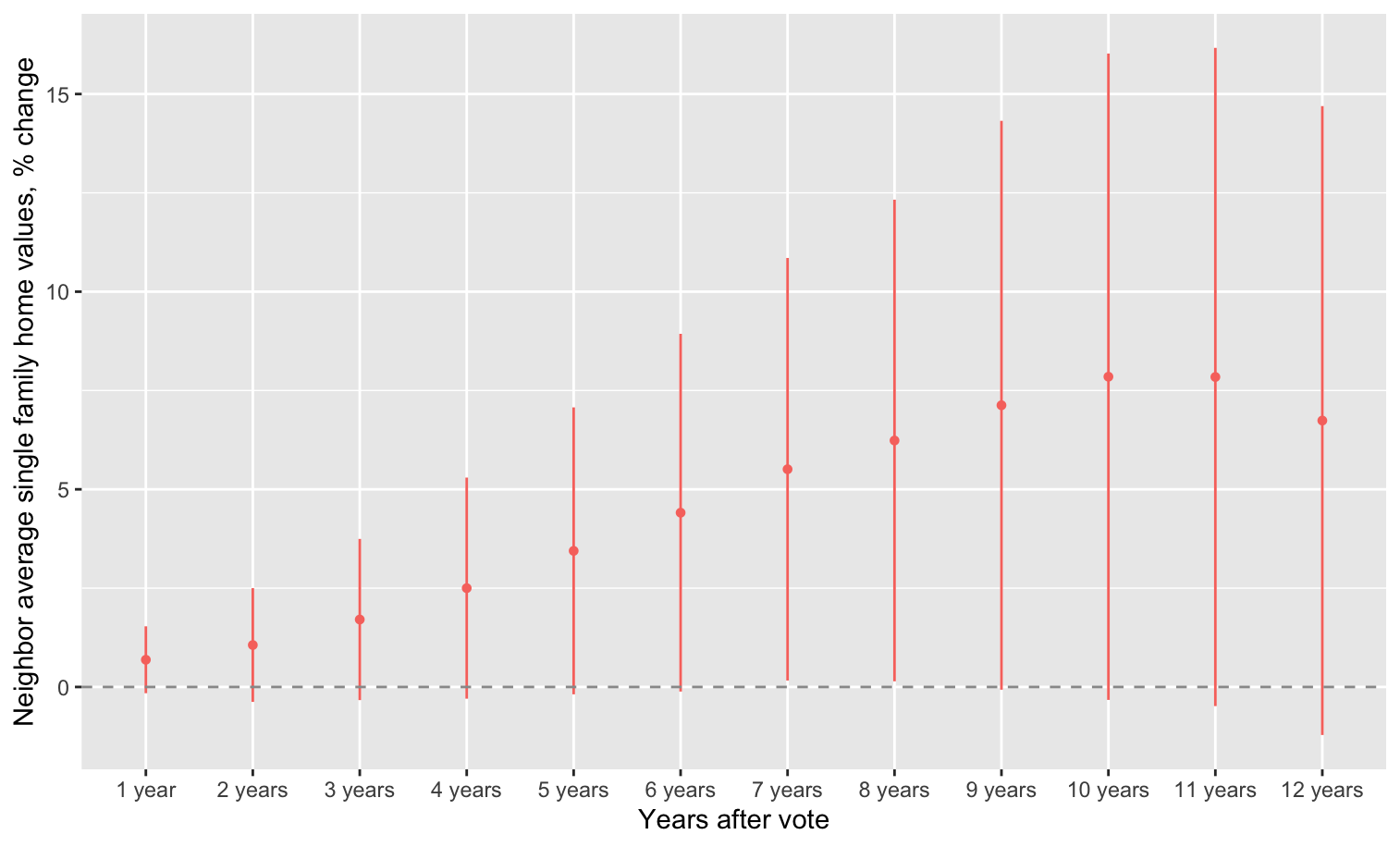}
    \label{fig:fuzzyRDD_nbravgfamvalcumulative}
\end{figure}

The results suggest that at a local level there are significant spillovers from municipal spending. While I eliminate spillovers coming through shared school districts, there are other services, such as sanitation, that are sometimes shared by neighboring municipalities. Moreover, changes in the composition of the population in one municipality may affect the preferences for neighboring municipalities, as in the large literature on assortative location choice. 

It appears that municipalities also respond to the changes in the tax levies of their neighbors. Figure \ref{fig:nbrlevypercentchange} shows the discontinuity in the tax bill growth rate in neighboring municipalities associated with a successful referendum (analogous to Figure \ref{fig:levypercentchange}). 

\begin{figure}[ht]
    \centering
    \caption{Discontinuity in levy growth in neighboring municipalities}
    \includegraphics[scale = .25]{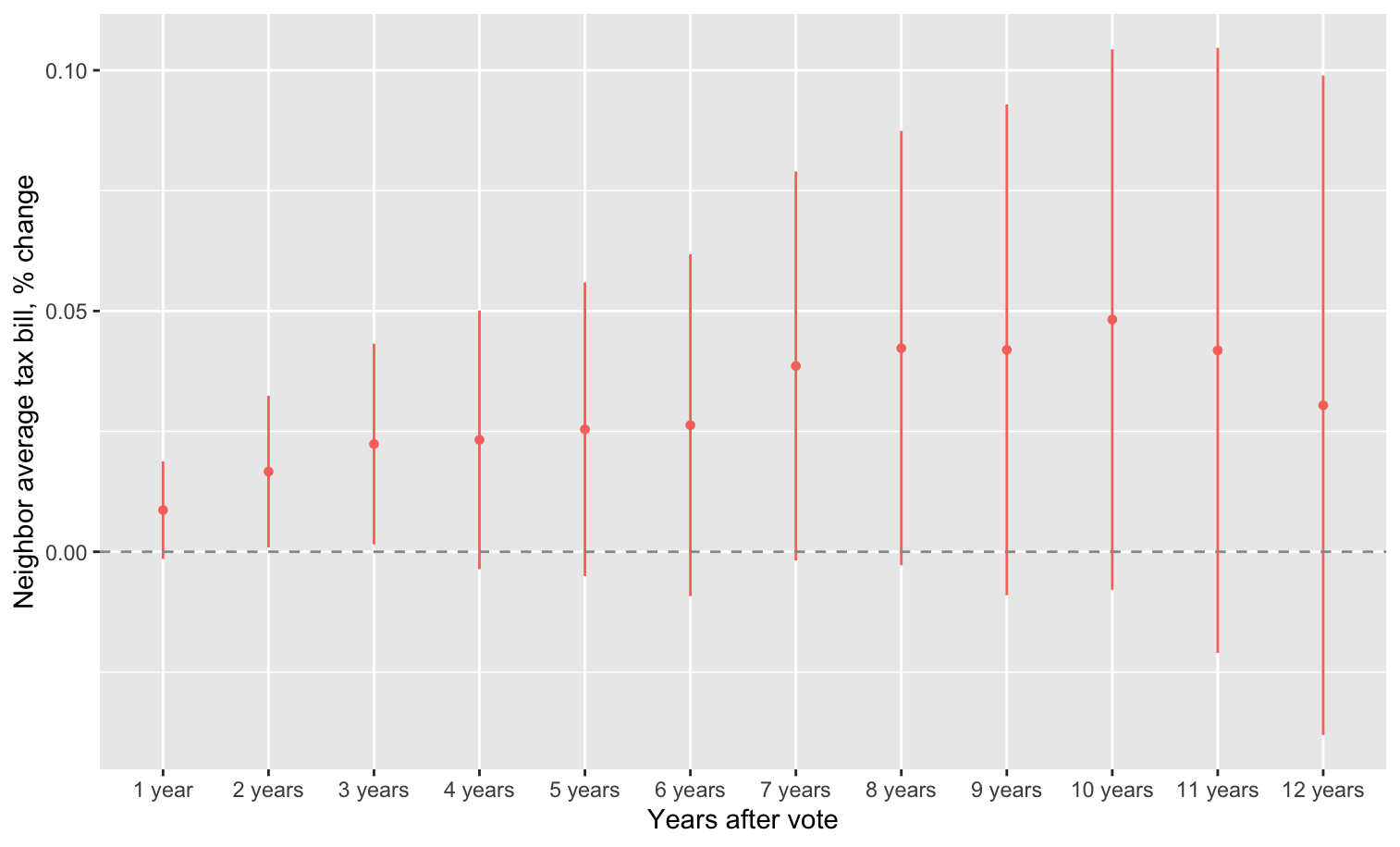}
    \label{fig:nbrlevypercentchange}
\end{figure}

Taxation in neighboring municipalities appears to respond positively to a successful referendum. The estimated discontinuity in levy growth for neighboring municipalities is displayed in Figure \ref{fig:nbrlevypercentchange}. There is some evidence that there is also an increase in the number of referenda held in neighboring municipalities, although the effects are not statistically significant. 

\subsection{Discussion}\label{sec:empiricaldiscussion}

Overall, I find some direct and indirect evidence of income sorting. The estimated effect of taxation levels on income per capita, corresponding to the model in (\ref{eq:incomePC}) and summarized in Figure \ref{fig:incomePC}, show a positive relationship between the levy and incomes. I prefer to be cautious in interpreting these results however, given that the effect disappears when I control for differences in initial conditions between the win and loss groups by estimating a model with percent changes rather than levels. 

The stronger evidence of income sorting comes from home values. Home sales prices should provide a more precise measure of marginal changes to the population, as the reflect the characteristics of new arrivals rather than the population mean. The effect on sales prices translates into higher assessed values. If find that despite the fact that higher taxes should be to some extend capitalized into home prices, leading to lower prices, increasing the growth in the tax bill leads to higher growth in home values. This is true for both contemporaneous (Figure \ref{fig:fuzzyRDD_avgfamval}) and lagged (Figure \ref{fig:fuzzyRDD_avgfamvalcumulative}) effects. 

Finally, I also study the effect of tax levies on neighboring municipalities. I find that increases in tax levels lead to both higher taxes and higher home values in neighboring municipalities. This may be because of some combination of unaccounted for spillovers in public goods expenditure, assortative moving, and strategic responses by other municipalities. The results suggest that the negative externalities of increased taxation may be less localized that we might expect ex-ante. Immediate neighbors may benefit, at the expense of more distant municipalities. The fact that a changes in a municipality's tax levy lead to similar changes in those of its neighbors is evidence of strategic responses by municipal governments. 

An interesting topic for future empirical work, suggested by the model and highlighted by the observed effects on neighboring municipalities, would be to better quantify the fiscal externalities of municipal expenditure. This could be done by looking directly at the effect of measurable features of local public goods, such as expenditure per pupil in schools, on tax levels and property values in neighboring districts. It would also be useful to study which types of expenditure have a greater impact on own and neighboring property values. Heterogeneity of these effects would likely be associated with distortions in the bundle of public goods provided by each municipality. Finally, the strategic behavior of municipalities with regards to taxation should depend on the fraction of renters versus owners. While both renters and owners dislike higher property taxes, assuming that the taxes are to some extent passed on from landowners to tenants, owners benefit from the increase in home values resulting from income sorting. Thus there should be less incentive to increase taxation in districts with a high proportion of renters. On the other hand, property taxes may be less salient for renters if they are passed through in the rental rate, rather than paid directly.

\section{Conclusion}

This paper analyses a dynamic model of local public goods provision. In equilibrium, competition by districts for wealthier residents leads to excessive taxation. The central government improves the welfare of the wealthiest districts by capping their expenditure level, mitigating the negative externalities of competition for high income residents. This surplus can the be transferred to low income districts, either through direct monetary transfers or by allowing migration flows to reduce geographic income disparity and increase the tax base in the low income districts. In either case, these districts will be able to increase expenditure on education. 

The fact that all districts benefit from such policies makes them politically feasible. No additional revenue sources are required, and interference in the local control of education funding is minimal. There is flexibility in determining which districts have schools that are over-funded, from the point of view of the policy maker, and should therefore be targeted for tax relief. The surplus generated by this tax relief can be flexibly divided among the districts, depending on the preferences of the central authority and political exigencies. 

The key driver of inefficiency in local public goods provision, which in turn implies that tax caps can be Pareto improving, is the sorting of richer new arrivals into districts with higher expenditure. I use data from Massachusetts to examine empirically whether this externality exists. I use a regression discontinuity design, exploiting regulation that requires municipalities to hold referenda in order to make certain increases in property taxation. I find some direct evidence that increasing the residential tax bill increases the mean household income. There is also strong evidence that increases in taxes increase home values. I interpret this as indirect evidence of income sorting. Surprisingly, I find that increasing taxation in a given municipality also increases home values in neighboring municipalities. This suggests that the negative externalities of expenditure are not as local as one might expect. It appears that the story is one of regional, rather than municipal competition. Nonetheless, tax caps may be an effective way to reduce negative fiscal externalities at the state level, while increasing funding for the poorest districts.

\newpage

\appendix

\section{Empirical}\label{app: empirical results}

\subsection{Percent change in levy}
The plots display the relationship between the margin of ``Yes'' votes (horizontal axis) and the percent change in the average single family tax bill (vertical axis) over different time horizons. It is easy to see that there is under-reporting of losses. An important identifying assumption, as stated above, is that the decision to report a loss, when there is not also a win that needs to be reported, is independent of municipality characteristics.

\begin{figure}[ht]
    \centering
    \includegraphics[scale = .3]{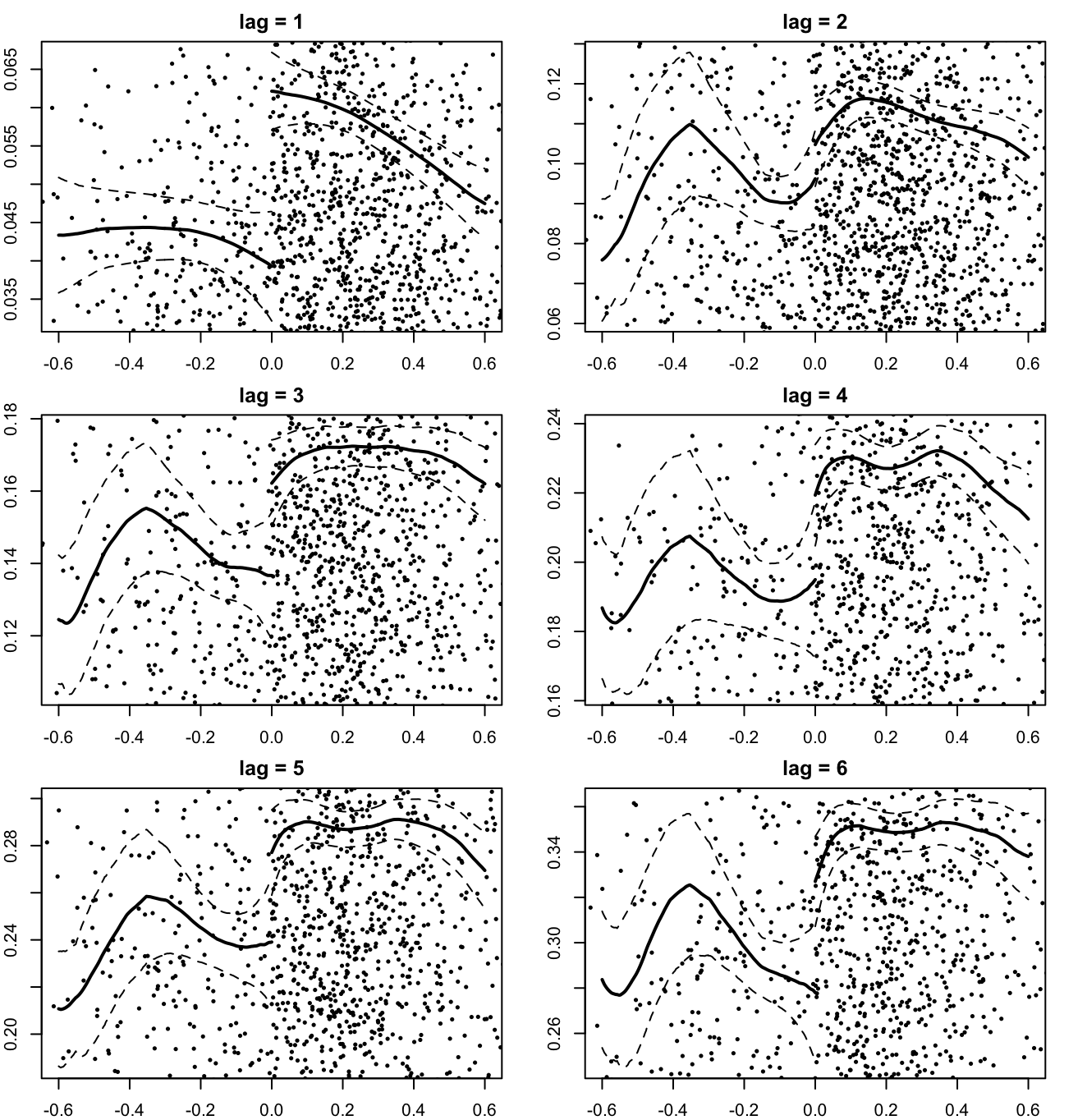}
    \label{fig:scatter16_levypercentchange}
\end{figure}

\begin{figure}[ht]
    \centering
    \includegraphics[scale = .3]{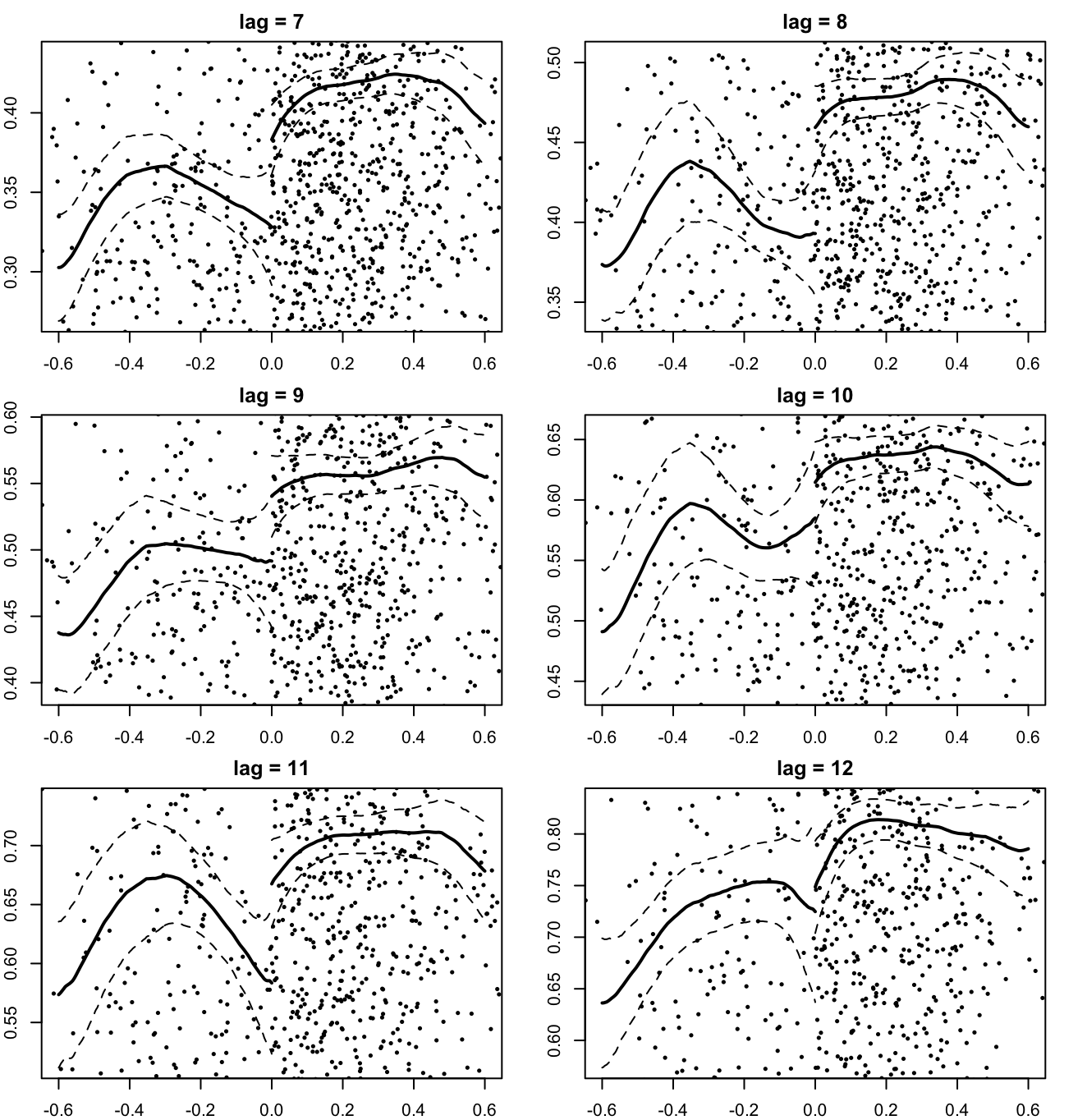}
    \label{fig:scatter712_levypercentchange}
\end{figure}

\newpage

\section{Omitted proofs} \label{app:proofs}

\noindent\textbf{Lemma \ref{lem0.1}}
\begin{proof}
Consider two locations with location qualities $h'' > h'$. Let $z'', z'$ be the net payment and price vectors, and $m'', m'$ the PDVs. Since $h'' > h'$ it must be that $m' > m''$.  

Monotonicity of the equilibrium will follow from the single crossing property of individual preferences over location quality and PDV: if a low type prefers a high location with low PDV then so does a high type. To show that single-crossing between $h$ and $w$ holds it is sufficient to show that 
\begin{equation}\label{eq0.2}
    V(w'', z'') - V(w', z'') \geq V(w'', z') - V(w', z'). 
\end{equation}
By the envelope theorem $V$ is differentiable in $w$ almost everywhere, with derivative given by
\begin{equation*}
    V_1(w,z) = u'\left(w - p_1 - \tau - b(w,z) \right) +  u'\left(w + p_2 + b(w,z)\right).
\end{equation*}
Then (\ref{eq0.2}) can be written as
\begin{equation*}
\int_{w'}^{w''} V_1(w,z'') dw \geq \int_{w'}^{w''} V_1(w,z') dw.    
\end{equation*}
This will hold if $V_1(w,z'') \geq V_1(w,z')$ for all $w \in [w',w'']$. The first order condition for $b$ implies that 
\begin{equation*}
   u'\left(w - p_1 - \tau - b(w,z) \right) = (1+r) u'\left(w + p_2 + (1+r)b(w,z) \right).
\end{equation*}
Using this FOC, $V_1(w,z) = (2 + r) \cdot u'(w + p_2 + (1+r)b(z,w))$, so $V_1(w,z') > V_1(w,z'') $ iff $u'(w + p_2 + (1+r)b(z',w)) > u'(w + p_2 + (1+r)b(z'',w))$. This follows immediately from concavity of $u$ and optimality of savings. 
\end{proof} 

\vspace{3mm}
\noindent \textbf{Proposition \ref{prop1}}
\begin{proof}
Let $q^j(w)$ be the home choice of type $w$ in district $j$, as defined in the text.

\textit{Claim 1:}  $q^B(w) + s(e^B) \geq q^A(w) + s(e^A)$ whenever $q^B(w) < \Bar{q}^B$ and $w$ locates on side $A$ with positive probability. Suppose some type $w$ locates in district $A$ in a home of quality $q^A(w)$ with PDV $m^A(w)$. Let $w' > w$ be a type that locates in district $B$, which must exist since $q^B(w) < \Bar{q}^B$ and the equilibrium is monotone by Lemma \ref{lem0.1}. By Lemma \ref{lem0.1} we must have $q^B(w') + s(e^B) \geq s(e^A) + q^A(w)$. If $q^B(w) + s(e^B) < q^A(w) + s(e^A)$ markets won't clear, since no types below $w$ can go to A district locations above $q^A(w)$. This proves the claim. 

\textit{Claim 2:} the converse to Claim 1; $q^B(w) + s(e^B) \leq q^A(w) + s(e^A)$ whenever $w$ locates on side $B$ and $q^A(w) > \und{q}^A$. The proof is symmetric to that of Claim 1. 

\textit{Part i.} This follows immediately from Claim 1. Let $w$ be the type that locates at $\und{q}^A$ on side $A$. The claim implies that only types below $w$ can locate on side B in homes of quality $q < s(e^A) - s(e^B) + \und{q}^A$. Since markets must clear (all homes must be occupied), this means that $w \geq w_*$. \textit{Part ii} follows similarly from Claim 2.

\noindent \textit{Part iii:}

\textit{Claim 3:} for all $w \in (w_*,w^*)$, $q^A(w) + s(e^A) = q^B(w) + s(e^B)$. Let $w'$ be any type above $w_*$ that locates in district $B$. Suppose $q^A(w') + s(e^A) > q^B(w') + s(e^B)$, that latter of which is greater than $\und{q}^A$ by the definition of $w_*$. By monotonicity of the equilibrium and market clearing, there exists a type $w < w'$ that locates in district $A$ with $q^A(w) \in (q^B(w') + s(e^B) - s(e^A), q^A(w'))$. But then, $q^B(w) + s(e^B) \geq q^A(w) + s(e^A) > q^B(w') + s(e^B)$, where the first inequality follows from Claim 1 and the second by assumption. By monotonicity, any type $w'' > w$ that locates on side $B$ must have $q^B(w'') \geq q^B(w)$. But then markets can't clear. Therefore $q^A(w') + s(e^A) = q^B(w') + s(e^B)$ for any such $w'$. A symmetric argument shows that $q^A(w') + s(e^A) = q^B(w') + s(e^B)$ for any $w' < w^*$ that locate in district $A$. 
\end{proof}

\vst
\noindent \textbf{Corollary \ref{cor:indifference}}
\begin{proof}
I wish to show that all $w \in (w_*, w^*)$ are indifferent between the two districts. Suppose such a $w$ strictly prefers district $A$. Then by Lemma \ref{lem0.2}, there exists $\varepsilon >0$ such that all types in $(w - \varepsilon, w + \varepsilon)$ strictly prefer district $A$. But then Claim 3, there is an interval of houses on side $B$, $(q^A(w - \varepsilon) + s(e^A) - s(e^B), q^A(w + \varepsilon) + s(e^A) - s(e^B))$, that cannot be occupied without violating montonicity, contradicting market clearing. The claim about payments is immediate.
\end{proof}

\vspace{3mm}
\noindent \textbf{Lemma \ref{lem3.5}}
\begin{proof}
Suppose that there are two different equilibria with the same expenditure gap, and thus the same location allocation, and with payment-price vectors $\hat{z}, \Tilde{z}$ and associated value functions $\hat{U},\hat{V}$ and $ \Tilde{U}, \Tilde{V}$. Let $s(\hat{e}^A) - s(\tilde{e}^A) = \Delta$ be the expenditure difference. 

\textit{Claim 1.} First I claim that that there cannot exist types $w'' > w'$ such that $\hat{V}(w') = \Tilde{V}(w')$ and $\hat{V}(w) > \Tilde{V}(w)$ for all $w \in (w',w'']$ (a symmetric argument shows that we cannot have $\hat{V}(w') = \Tilde{V}(w')$ and $\hat{V}(w) < \Tilde{V}(w)$ for all $w \in (w',w'']$). Suppose such types do exist. Since locations assignments are the same in both equilibria $\hat{V}(w') = \Tilde{V}(w')$ iff $\hat{U}(w') - \Tilde{U}(w') = \Delta$. By strict concavity of $u$ and $v$ and the FOC for $(b_1,b_2)$ we have $V(w,\hat{z}(w)) > V(w,\Tilde{z}(w)) \Leftrightarrow V_1(w,\hat{z}(w)) < V_1(w,\Tilde{z}(w))$. Then the envelope condition implies that $\hat{U}(w'') - \Tilde{U}(w'') <  \Delta$. But this holds iff $\hat{V}(w'') < \tilde{V}(w'')$, contradicting our assumption. We conclude that such a $w'',w'$ cannot exist. 

By assumption $V(\und{w}) = \und{V}$ in any equilibrium (because of the subsidized housing outside option). The above claim implies that any interval $[\und{w}, w'']$ contains points $w_1, w_2$ such that $\tilde{V}(w_1) \geq \hat{V}(w_1)$ and $\hat{V}(w_2) \geq \tilde{V}(w_2)$. Suppose $w_2 > w_1$ (a symmetric argument applies if $w_1 > w_2$). Since $V$ is continuous there must exist a non-degenerate interval $[w', w_2]$ such that $\hat{V}(w') = \Tilde{V}(w')$ and $\hat{V}(w) > \Tilde{V}(w)$ for all $w \in (w',w_2]$. But this contradicts Claim 1. Therefore $\hat{V}(w) = \tilde{V}(w)$ for all $w$.
\end{proof}

\vspace{3mm}
\noindent \textbf{Proposition \ref{prop2}}
\begin{proof}
Let $U^j$ and $\hat{U}^j$ be the value functions under $e^A$ and $\hat{e}^A$ respectively. Let $m^j, \hat{m}^j$ be the lump-sum equivalents of $z^j$ and $\hat{z}^j$. By Corollary \ref{cor1} it is without loss to consider $\hat{e}^B = e^B$. 

I first want to show that payment-price vectors are worse in district $A$ under $\hat{e}^A$. There are two cases to consider: $\und{q}^A + s(\hat{e}^A) \geq \und{q}^B + s(\hat{e}^B)$ and $\und{q}^A + s(\hat{e}^A) < \und{q}^B + s(\hat{e}^B)$.

\noindent \textit{Part 1.} Assume $\und{q}^A + s(\hat{e}^A) \geq \und{q}^B + s(\hat{e}^B)$. 

Define the wealth of the individual moving to location $q$ in district $A$ as a function of district $A$ expenditure as $w^A(q,e)$. If $\und{q}^A + s(\hat{e}^A) \geq \und{q}^B + s(\hat{e}^B)$ then this is defined implicitly by 
\begin{equation*}
    Q^A(q) + Q^B(q + s(e^A) - s(e^B)) = F(w^A(q,e)).
\end{equation*}
Since $Q^A,Q^B$ and $F$ have strictly positive densities, and $s$ has a strictly positive derivative by assumption, implicit differentiation yields $w_2(q,e) > 0$. Furthermore, since the densities of $Q^A,Q^B$ and $F$ are bounded away from zero, there exists $\delta > 0$ such that for any home in district $A$, $w^A(q,e'') - w^A(q,e') \geq \delta(e''-e')$ for all $e'' > e'$. 

\vspace{3mm}
\noindent \textit{Part 1.A.} Assume that $\und{q}^A + s(e^A) \geq \und{q}^B - s(e^B)$. Define $\tilde{q} = \inf \{q : \hat{m}_q^A \geq m_q^A - \delta(\hat{e}^A - e^A)\}$. 

Suppose $\tilde{q} > \und{q}^A$. Then $ \hat{m}_{\tilde{q}}^A = m_{\tilde{q}}^A - \delta(\hat{e}^A - e^A)$, by continuity of the value functions $U^A$,$\hat{U}^A$ and the location allocations $q^A, \hat{q}^A$. The envelope condition for $U$ at $w^A_{\tilde{q}}$ holds if and only if 
\begin{align}\label{eq4.9}
    &\lim_{w' \rightarrow w^A_{\tilde{q}}} \dfrac{V(w',m(w')) - V(w', m(w^A_{\tilde{q}}))}{w' - w^A_{\tilde{q}}} + \dfrac{q^A(w') - q^A(w^A_{\tilde{q}})}{w' - w^A_{\tilde{q}}} = 0 \notag\\
    &\Leftrightarrow \lim_{w' \rightarrow w^A_{\tilde{q}}} \left(\dfrac{V(w',m(w')) - V(w', m(w^A_{\tilde{q}}))}{q^A(w') - q^A(w^A_{\tilde{q}})} + 1\right) \left( \dfrac{q^A(w') - q^A(w^A_{\tilde{q}})}{w' - w^A_{\tilde{q}}} \right) = 0
\end{align}
$F$ and $Q^A, Q^B$ are differentiable with strictly positive derivatives, so $q^A$ is also differentiable and has a strictly positive derivative. So (\ref{eq4.9}) holds iff
\begin{equation}\label{eq4.15}
  \lim_{w' \rightarrow w^A_{\tilde{q}}} \left(\dfrac{V(w',m(w')) - V(w', m(w^A_{\tilde{q}}))}{q(w') - q(w^A_{\tilde{q}})} + 1\right) = 0.
\end{equation}
We can rewrite (\ref{eq4.15}) as
\begin{equation}\label{eq4.16}
  \lim_{\varepsilon \rightarrow 0} \dfrac{V(w^A_{\tilde{q} - \varepsilon},m(w^A_{\tilde{q} - \varepsilon})) - V(w^A_{\tilde{q} - \varepsilon}, m(w^A_{\tilde{q}}))}{\varepsilon} = 1.
\end{equation}

Consider the right limit of (\ref{eq4.16}) as $\varepsilon$ goes to zero from above. By definition of $\tilde{q}$, $\hat{m}_q^A < m_q^A - \delta (\hat{e}^A - e^A)$ for all $q < \tilde{q}$.

Then there exists a $\kappa > 0$ such that for any $\varepsilon > 0$
\begin{align}
    V(w^A_{\tilde{q}- \varepsilon} &,m(w^A_{\tilde{q} - \varepsilon})) - V(w^A_{\tilde{q}- \varepsilon} , m(w^A_{\tilde{q}})) = \int\limits_{m(w^A_{\tilde{q}})}^{m(w^A_{\tilde{q}- \varepsilon} )} V_2(w^A_{\tilde{q}- \varepsilon} , x)dx \notag\\
    &\geq \kappa(m(w^A_{\tilde{q}- \varepsilon} ) - m(w^A_{\tilde{q}})) + \int\limits_{m(w^A_{\tilde{q}})}^{m(w^A_{\tilde{q}- \varepsilon} )} V_2(w^A_{\tilde{q} - \varepsilon} + \delta(\hat{e}^A - e^A), x - \delta(\hat{e}^A - e^A))dx \notag\\
    &\geq \kappa(m(w^A_{\tilde{q}- \varepsilon} ) - m(w^A_{\tilde{q}})) + \int\limits_{\hat{m}(\hat{w}^A_{\tilde{q}})}^{\hat{m}(\hat{w}^A_{\tilde{q}- \varepsilon} )} V_2(w^A_{\tilde{q} - \varepsilon} + \delta(\hat{e}^A - e^A), x - \delta(\hat{e}^A - e^A) + m(w^A_{\tilde{q}}) - \hat{m}(\hat{w}^A_{\tilde{q}}))dx \notag\\
     &= \kappa(m(w^A_{\tilde{q}- \varepsilon} ) - m(w^A_{\tilde{q}})) + \int\limits_{\hat{m}(\hat{w}^A_{\tilde{q}})}^{\hat{m}(\hat{w}^A_{\tilde{q}- \varepsilon} )} V_2(w^A_{\tilde{q} - \varepsilon} + \delta(\hat{e}^A - e^A), x)dx \notag\\
     &\geq \kappa(m(w^A_{\tilde{q}- \varepsilon} ) - m(w^A_{\tilde{q}})) + \int\limits_{\hat{m}(\hat{w}^A_{\tilde{q}})}^{\hat{m}(\hat{w}^A_{\tilde{q}- \varepsilon} )} V_2(\hat{w}^A_{\tilde{q} - \varepsilon} , x )dx \notag\\
     &= \kappa(m(w^A_{\tilde{q} - \varepsilon}) - m(w^A_{\tilde{q}})) + V(\hat{w}^A_{\tilde{q} - \varepsilon},\hat{m}(\hat{w}^A_{\tilde{q}- \varepsilon} )) - V(\hat{w}^A_{\tilde{q}- \varepsilon}  , \hat{m}(\hat{w}^A_{\tilde{q}})). \label{eq4.10}
\end{align}
The first inequality follows from two facts. First, adding some amount $y$ to an individual's wealth and subtracting $y$ from their lump-sum payment-price equivalent increases the total budget constraint in their inter-temporal savings problem by $y/(1+r)$. Second, strict concavity of $u$ and $v$ implies that $V_2$ decreases when the total budget increases. Moreover $\kappa$ can be chosen uniformly over $\varepsilon$, since the slopes of $u'$ and $v'$ are bounded away from 0. 

The second inequality holds since $\hat{m}_q^A - m_q^A < -\delta (\hat{e}^A - e^A)$ for all $q < \tilde{q}$ and $\hat{m}_{\tilde{q}}^A - m_{\tilde{}}^A = -\delta (\hat{e}^A - e^A)$ implies $\hat{m}(\hat{w}^A_{\tilde{q} - \varepsilon}) - \hat{m}(\hat{w}^A_{\tilde{q}}) \leq m(w^A_{\tilde{q} - \varepsilon}) - m(w^A_{\tilde{q}})$. 

Rewriting the envelope condition for $\hat{U}$ at $\hat{w}_{\tilde{q}}$ as in (\ref{eq4.16}), and using (\ref{eq4.10}), we can see that the envelope condition of $\hat{U}$ will be violated if

\begin{equation}\label{eq4.13}
    \limsup_{\varepsilon \rightarrow 0} \dfrac{m(w_{\tilde{q}- \varepsilon}^A ) - m(w_{\tilde{q}}^A)}{\varepsilon} > 0.
\end{equation}

To see that this holds, rewrite (\ref{eq4.16}) as 
\begin{equation}\label{eq4.17}
    1 = \lim_{\varepsilon \rightarrow 0} \dfrac{V(w^A_{\tilde{q}- \varepsilon} ,m(w^A_{\tilde{q} - \varepsilon})) - V(w^A_{\tilde{q}- \varepsilon}, m(w^A_{\tilde{q}}))}{m(w^A_{\tilde{q} - \varepsilon} ) - m(w^A_{\tilde{q}})}\times \dfrac{m(w^A_{\tilde{q}- \varepsilon}) - m(w^A_{\tilde{q}})}{\varepsilon}
\end{equation}

Monotonicity requires that $m(w_{\tilde{q}- \varepsilon}^A ) - m(w_{\tilde{q}}^A) > 0$. Suppose
\begin{equation*}
    \limsup_{\varepsilon \rightarrow 0} \dfrac{m(w_{\tilde{q}- \varepsilon}^A ) - m(w_{\tilde{q}}^A)}{\varepsilon} = 0.
\end{equation*}
Then by the squeeze theorem the limit exists and is equal to zero. But then (\ref{eq4.17}) holds only if
\begin{equation*}
    \lim_{\varepsilon \rightarrow 0} \dfrac{V(w^A_{\tilde{q}- \varepsilon} ,m(w^A_{\tilde{q} - \varepsilon})) - V(w^A_{\tilde{q}- \varepsilon}, m(w^A_{\tilde{q}}))}{m(w^A_{\tilde{q} - \varepsilon} ) - m(w^A_{\tilde{q}})}\times \dfrac{m(w^A_{\tilde{q}- \varepsilon}) - m(w^A_{\tilde{q}})}{\varepsilon} = + \infty
\end{equation*}
But fixing $\bar{\varepsilon} >0$, concavity of $u$ and $v$ implies that this limit is bouned above by $V_2(w^A_{\tilde{q} - \bar{\varepsilon}})$, which exists by the envelope theorem applied to the savings problem and is bounded by our assumptions on $u$ and $v$ and the fact that $m_q > 0$ for all $q$. This proves that (\ref{eq4.13}) is satisfied, so the envelope condition for $\hat{U}$ cannot hold. Thus we conclude that $\tilde{q} > \und{q}^A$  cannot hold. 

\vspace{3mm}
\noindent \textit{Part 1.B.} Continue to assume that $\und{q}^A + s(e^A) \geq \und{q}^B + s(e^B)$. Let $q^*, w^*$ be the cutoff home and type such that under $e^A$ all types below $w^*$ move to homes in district $B$ with qualities below $q^*$, and $w^*$ is indifferent between $q^*$ in district $B$ and the lowest quality house in district $A$. That is $s(e^A) + \und{q}^A = s(e^B) + q^* $ and $w^* = w^B_{q^*} = w^A_{\und{q}^A}$. Under $\hat{e}^A$ all types below $w^*$ move to the same location, and by Lemma \ref{lem3.6} they therefore receive the same money values. Let $m^j(q,e)$ be the lump-sum equivalent for location $q$ in district $j$ when district $A$ expenditure is $e$, holding district $B$ expenditure fixed. Monotonicity of the location allocation requires that for all $e' > e^A$, type $w^*$ strictly prefer side $B$ to side $A$ under $e'$, i.e. 
\begin{align*}
    \und{q}_A + s(e') + V(w^*, m^A(\und{q}^A,e')) &< q^* + s(e^B) + V(w^*, m^B(q^*, e')) \\
    & = \und{q}^A + s(e^A) + V(w^*, m^A(\und{q}^A, \hat{e}^A))
\end{align*}
where the equality follows from the definition of $q^*$ and Proposition \ref{prop1}. So it must be that $m^A(\und{q}^A, e') < m^A(\und{q}^A, e^A)$. Moreover  
\begin{equation*}
    \limsup_{\varepsilon \searrow 0}  \dfrac{V(w^*, m^A(\und{q}^A,e^A + \varepsilon)) - V(w^*, m^A(\und{q}^A,e^A))}{\varepsilon} \leq \limsup_{\varepsilon \searrow 0} \dfrac{s(e^A) - s(e^A + \varepsilon)}{\varepsilon} = - s'(e^A) < 0.
\end{equation*}
This holds only if 
\begin{equation*}
    \limsup_{\varepsilon \searrow 0} \dfrac{m^A(\und{q}^A,e^A + \varepsilon) - m^A(\und{q}^A,e^A)}{\varepsilon} < 0.
\end{equation*}
Then exists $\gamma > 0$ such that $\hat{m}^A_{\und{q}^A} < m^A_{\und{q}^A} - \gamma(\hat{e}^A - e^A)$ for all $\hat{e}^A > e^A$. If $\gamma > \delta$ then this shows that $\tilde{q} \neq q^A$, so we are done. Otherwise replace $\delta$ with $\gamma$ in Part1.B to complete the proof. 

\vspace{3mm}
\noindent\textit{Part 2.} Suppose $s(\hat{e}^A) + \und{q}^A < s(\hat{e}^B) + \und{q}^B$. By an almost identical argument to Part 1.B, we can show that there exists $\gamma > 0$ such that $\hat{m}^A_{\und{q}^A} < m^A_{\und{q}^A} - \gamma(\hat{e}^A - e^A)$ for all $\hat{e}^A > e^A$. Then the proof proceeds as in Part 1.A.

\noindent\textit{Part 3.} The conclusions regarding prices in district B follow from symmetric arguments. In this case we the idea is that if payment-price vectors do not improve higher types will prefer to deviate to lower quality houses. 
\end{proof}

\vspace{3mm}
\noindent \textbf{Proposition \ref{prop5.8}}
\begin{proof}
The weak comparative statics result will follow from the following weak increasing differences property: for all $e' > e^A_1$, we have $\Pi(e',\tau^A_1(e')) - \Pi(e^A_1, \tau^A_1(e^A_1)) \geq \tilde{\Pi}(e', \tilde{\tau}^A_1(e')|e_1^A) - \tilde{\Pi}(e^A_1, \tilde{\tau}^A_1(e^A_1)|e^A_1)$.

First, notice that $\tau^A_1(e^A_1) = \tilde{\tau}^A_1(e^A_1) = \tau^A_1$ and $\Pi(e^A_1, \tau^A_1) = \tilde{\Pi}(e^A_1, \tau^A_1|e_1^A)$, so this is equivalent to $\Pi(e',\tau^A_1(e')) \geq \tilde{\Pi}(e', \tilde{\tau}^A_1(e'|e^A_1)|e^A_1)$. By Proposition \ref{prop2} $m^A(q|e') \geq m^A(q|e^A_1)$ for all $q$. So any tax schedule that yields revenue of $e'$ is strictly better for the district's perspective under $m^A(\cdot|e')$ than under $m^A(\cdot|e^A_1)$. 

By the envelope theorem $\tilde{\Pi}(e, \tilde{\tau}^A_1(e|e^A_1)|e^A_1)$ is differentiable in $e$.  The strict comparative statics result follows from the conditions in Proposition \ref{prop2} under which 
\begin{equation*}
    \limsup_{\varepsilon \searrow 0} \dfrac{m^A(q|e^A_1 + \varepsilon) - m^A(q|e^A_1)}{\varepsilon} < 0
\end{equation*}
\end{proof}

\vst
\noindent\textbf{Proposition \ref{lem6.9}}
\begin{proof}
Let $\tilde{m}^1_q(e_1)$ be the lump-sum equivalent for home $q$ in district $1$ in the fixed gap problem when district 1 chooses first period expenditure $e_1$, and let $m^1_q(e_1)$ be the lump-sum equivalent when district $1$ chooses $e_1$ and all other districts are fixed at their equilibrium expenditure levels. Repeatedly applying Proposition \ref{prop4}, for $e_1 > e^1_1$ we have $\tilde{m}^1_1(e_1) \geq m^1_1(e_1)$, with equality if and only if $s(\tilde{e}^1_1) + \und{q}^1 \geq s(e^j_1) + \bar{q}^j$ for all $j \in Z$ or $s(e^1_1) + \bar{q}^1 \leq s(e^j) + \und{q}^j$ for all $j \in Z$. This gives the desired single crossing property, as in the proof of Proposition \ref{prop5.8}. 
\end{proof}

\vst
\noindent\textbf{Proposition \ref{prop5}}
\begin{proof}
That districts in $Z$ are better off follows immediately from Proposition \ref{lem6.9}, given concavity of $\tilde{\Pi}(e_1,\tilde{\tau}(e_1))$. Proposition \ref{prop4} implies that districts not in $Z$ are made better off by the cap. 
\end{proof}

\subsection{Alternative district objectives}

Consider district governments that are in power for one period and care only about school quality, tax burdens, home values, and the wealth distributions of new residents; district $j$ payoffs are given by $K^j(e^j, \tau^j, p^j, \Gamma^j)$, where $\Gamma^j$ is the wealth distribution of new arrivals to district $j$. In each period districts simultaneously choose expenditure levels and tax schedules. Tax schedules map home qualities and home prices to payments. Given these objects, individuals make their location choices and prices and district wealth distributions are realized. The district's objective is assumed to satisfy natural monotonicity and continuity conditions

\vspace{3mm}
\noindent \textbf{District objective assumptions.} $K^j$ is strictly increasing in $e^j$, decreasing in $\tau^j$ and strictly increasing in $p^j$ (with the uniform dominance partial order), and weakly increasing in $\Gamma^j$ (with the FOSD partial order). Moreover, $K^j$ is (sup-norm) continuous in $\tau^j$. 
\vspace{3mm}

Home prices depend on both current and future expenditure levels and tax schedules in each district, but are independent of the actions taken by districts in the past. Given this observation, I make a natural independence assumption on equilibria: district actions are history independent. 

As discussed in the previous section, higher expenditure levels lead to a richer population of new arrivals and worse payment-price vectors for every household. In the equilibrium of the game between districts, this will mean higher home prices. When districts want higher home values and/or richer residents this will push them to set higher expenditures in equilibrium than would be optimal given the wealth distribution of the new arrivals that the districts receive in equilibrium.   

To formalize this discussion, consider the problem faced by the government of district $A$. I fix the strategy of district $B$ and suppress dependence on this in the notation that follows. Given the choice of $e^A$, the set of residents moving to district $A$ and the payment-price vectors for each house are uniquely determined. Let $m^A(q|e^A)$ be the PDV of the payment-price vector for house $q$. By Lemma \ref{lem3.5}, $m^j$ is uniquely determined by the expenditure gap, so that a district's optimal action also depends only on the expenditure gap. The price function $p^A(\cdot)$ is given by
\begin{equation}\label{eq6.1}
    p^A(q|e,\tau) = - \tau(q) + \dfrac{1}{1+r}p_*^A(q) - m^A(q|e).
\end{equation}
where $p_{*}^A(q)$ is the anticipated equilibrium price in the next period, which by the independence assumption on equilibria does not depend on the period 1 district actions. Fixing the actions of the district $B$ government, the objective of district $A$ can be written as  
\begin{equation*}
    \Pi^A(e, \tau) = K^A(e,\tau, p^A(\cdot|e,\tau), \Gamma(e))
\end{equation*}
where $p^A(q| e_1, \tau_1)$ is given by (\ref{eq5.1}) and $\Gamma^A(e^A,e^B)$ is the wealth distribution of new arrivals determined by expenditure levels $e^A$ and $e^B$ in districts $A$ and $B$ respectively. 

The objective $\Pi^A$ takes into account the effect of $e^A$ on the population of new arrivals and prices. Let $e_*^j$ be the equilibrium expenditure in district $j$. Consider instead fixing the at $\bar{e}, e_*^B$, for expenditures $\bar{e}, e_*^B$ in districts $A$ and $B$ respectively. This amounts to modifying the objective from $\Pi^A$ to
\begin{equation*}
    \tilde{\Pi}^A(e, \tau| \bar{e}) = K^A(e,\tau, p^A(\cdot|\bar{e},\tau), \Gamma(\bar{e})).
\end{equation*}
Let $\tilde{e}^A_1(\bar{e})$ be the solution to this problem. Given a choice of $e$, let $\tau^A(e)$ and $\tilde{\tau}^A(e|\bar{e})$ be the optimal tax schedules for collecting $e$ when the objectives are $\Pi^A$ and $\tilde{\Pi}^A(\cdot,\cdot|\bar{e})$ respectively. Let $e^A$ and $e^B$ be the equilibrium expenditure levels in the two districts. Under the stated assumptions it is easy to see that competition leads to over-taxation. 

\begin{lemma}\label{lem6.8}
$ e^A_1 \geq \tilde{e}^A_1(e_1^A)$, with equality if and only if $s(\tilde{e}^A_1(e^A_1)) + \und{q}^A \geq s(e^B_1) + \bar{q}^B$ or $s(e^A_1) + \bar{q}^A \leq s(e^B) + \und{q}^B$. The same conclusion holds switching the labels $A$ and $B$. 
\end{lemma}
\begin{proof}
The weak comparative statics result will follow from the following weak increasing differences property: for all $e' > e^A_1$, we have $\Pi(e',\tau^A_1(e')) - \Pi(e^A_1, \tau^A_1(e^A_1)) \geq \tilde{\Pi}(e', \tilde{\tau}^A_1(e')|e_1^A) - \tilde{\Pi}(e^A_1, \tilde{\tau}^A_1(e^A_1)|e^A_1)$.

First, notice that $\tau^A(e^A) = \tilde{\tau}^A(e^A) = \tau^A$ and $\Pi(e^A, \tau^A) = \tilde{\Pi}(e^A, \tau^A|e^A)$, so this is equivalent to $\Pi(e',\tau^A(e')) \geq \tilde{\Pi}(e', \tilde{\tau}^A(e'|e^A)|e^A)$. By Proposition \ref{prop2} $m^A(q|e') \geq m^A(q|e^A)$ for all $q$. Since future prices do not depend on current actions, this means that current home values are higher under $e'$ than under $e^A$ So any tax schedule that yields revenue of $e'$ is strictly better for the district's perspective under $m^A(\cdot|e')$ than under $m^A(\cdot|e^A)$. 

By the envelope theorem $\tilde{\Pi}(e, \tilde{\tau}^A(e|e^A)|e^A)$ is differentiable in $e$.  The strict comparative statics result follows from the conditions in Proposition \ref{prop2} under which 
\begin{equation*}
    \limsup_{\varepsilon \searrow 0} \dfrac{m^A(q|e^A + \varepsilon) - m^A(q|e^A)}{\varepsilon} < 0
\end{equation*}
\end{proof}

When district payoffs are quasi-concave, Pareto improving action profiles can be easily identified an implemented by an expenditure cap. An \textit{expenditure cap} is an upper bound on the level of expenditure a district can choose. 

\vspace{3mm}
\noindent \textbf{Q-concave assumption}. The function $e \mapsto \tilde{\Pi}^A(e, \tilde{\tau}^A(e|\bar{e})|\bar{e})$ is quasi-concave.
\vspace{3mm}

Districts' payoff as a function of $e$ will generally be quasi-concave when there are decreasing returns to expenditure and the district governments experience increasing costs of taxing the population.

\begin{proposition}\label{prop6.1}
For any expenditure caps $\bar{e}^A \in [\tilde{e}^A(e^A), e^A)$ and $\bar{e}^B \in [\tilde{e}^B(e^B), e^B)$ such that $\bar{e}^A - \bar{e}^B = e^A - e^B$, both districts are strictly better off. In this range the welfare of each district is strictly decreasing in the cap level. 
\end{proposition}

\newpage
\bibliography{property_tax}

\end{document}